\documentclass{llncs}

\let\terms\undefined
\usepackage[utf8]{inputenc}
\usepackage[T1]{fontenc}
\usepackage{wrapfig}
\usepackage{makecell}
\usepackage{amsmath}
\usepackage{amssymb}
\usepackage{array}
\usepackage{multirow}
\usepackage{multicol}
\usepackage{todonotes}
\usepackage{booktabs}
\usepackage{pifont}
\usepackage{tikz}
\usepackage{hyperref}
\usepackage{textcomp}
\usepackage[capitalise,nameinlink]{cleveref}
\usepackage{subcaption}
\usepackage{color}
\usepackage{float}
\usepackage{colortbl}
\usepackage{color}
\usepackage{ltl}
\usepackage{xspace}
\usepackage{varwidth}
\usepackage{stmaryrd}
\usepackage{listings}
\usepackage[deletedmarkup=xout]{changes}
\usepackage{algorithm}
\usepackage[noend]{algpseudocode}

\crefname{appendix}{Appendix}{Appendix}
\crefname{section}{Sec.}{Sec.}
\crefname{theorem}{Thm.}{Thm.}

\usepackage{amsthm}
\theoremstyle{definition}

\usepackage{etoolbox}
\preto\align{\par\small\noindent}
\expandafter\preto\csname align*\endcsname{\par\small\vspace{-2mm}}
\preto\equation{\par\small\noindent}
\expandafter\preto\csname equation*\endcsname{\par\small\vspace{-2mm}}

\usetikzlibrary{calc}
\usetikzlibrary{automata}
\usetikzlibrary{backgrounds}
\usetikzlibrary{decorations.pathreplacing}
\usetikzlibrary{shapes,arrows}
\usetikzlibrary{positioning}
\usetikzlibrary{shadows}
\usetikzlibrary{circuits.logic.US}
\usetikzlibrary{arrows}
\usepackage{pgf}

\newcommand{\nats}{\ensuremath{\mathbb{N}}}
\newcommand{\signal}{\ensuremath{s}}

\newcommand{\from}{\ensuremath{\colon}}
\newcommand{\fspace}[2]{\ensuremath{#2^{[#1]}}}
\newcommand{\dtime}{\ensuremath{\nats}}
\renewcommand{\to}{\ensuremath{\rightarrow}}
\newcommand{\values}{\ensuremath{\mathcal{V}}}
\newcommand{\inputs}{\ensuremath{\mathcal{I}}}
\newcommand{\outputs}{\ensuremath{\mathcal{O}}}

\newcommand{\cells}{\ensuremath{\mathbb{C}}\xspace}
\newcommand{\bool}{\ensuremath{\mathcal{B}}}

\newcommand{\cfm}{\ensuremath{\mathcal{M}}\xspace}
\newcommand{\branch}[2]{\ensuremath{#1 \! \wr \hspace{-0.6pt} #2}}
\newcommand{\term}{\ensuremath{\tau}}
\newcommand{\fterm}{\ensuremath{\term_{F}}\xspace}
\newcommand{\pterm}{\ensuremath{\term_{P}}\xspace}

\newcommand{\name}[1]{{\text{\texttt{#1}}}}
\newcommand{\const}[1]{\ensuremath{\text{\texttt{#1}}()}}
\newcommand{\terms}{\ensuremath{\mathcal{T}}\xspace}
\newcommand{\comp}{\ensuremath{\varsigma}\xspace}
\newcommand{\comps}{\ensuremath{\mathcal{C}}}

\newcommand{\inits}{\ensuremath{\textit{init}_{\hspace{1pt}\name{c}}}\xspace}
\newcommand{\initsE}{\ensuremath{\textit{init}_{\hspace{1pt}\name{s}_{\name{i}}}}\xspace}
\newcommand{\sep}{\ensuremath{\quad | \quad}}
\newcommand{\upd}[2]{\ensuremath{\llbracket \hspace{0.5pt} #1 \leftarrowtail \, #2 \hspace{0.5pt} \rrbracket}}
\newcommand{\sats}{\ensuremath{\;\vDash_{\!\langle \hspace{-1pt} \cdot \hspace{-1pt} \rangle}}}
\newcommand{\nsats}{\ensuremath{\;\nvDash_{\!\langle \hspace{-1pt} \cdot \hspace{-1pt} \rangle}}}
\newcommand{\set}[1]{\ensuremath{\{ #1 \}}}
\newcommand{\eval}{\ensuremath{\eta_{_{\!\langle \hspace{-1pt} \cdot \hspace{-1pt} \rangle\!}}\!}\xspace}
\newcommand{\evalid}{\ensuremath{\eta_{_{\!\langle \hspace{-1pt} \cdot \hspace{-1pt} \rangle_{\texttt{id}}\!}}\!}\xspace}
\newcommand{\impl}{\ensuremath{\rightarrow}}

\newcommand{\pterms}{\ensuremath{\terms_{\!P\hspace{-0.5pt}}}\xspace}
\newcommand{\fterms}{\ensuremath{\terms_{\!F\hspace{-0.5pt}}}\xspace}
\newcommand{\uterms}{\ensuremath{\terms_{\hspace{-0.5pt}U \!}}\xspace}
\newcommand{\utermsp}{\ensuremath{\terms_{\hspace{-0.5pt}U / \name{id}}}}
\newcommand{\functions}{\ensuremath{\mathcal{F}}\xspace}
\newcommand{\predicates}{\ensuremath{\mathcal{P}}\xspace}
\newcommand{\assign}[1]{\ensuremath{\langle #1 \rangle}}

\newcommand{\TSL}{\text{TSL}\xspace}

\newcommand{\inames}{\ensuremath{\mathbb{I}}\xspace}
\newcommand{\onames}{\ensuremath{\mathbb{O}}\xspace}
\newcommand{\pnames}{\ensuremath{\mathbb{P}}\xspace}
\newcommand{\fnames}{\ensuremath{\mathbb{F}}\xspace}

\newcommand{\vertices}{\ensuremath{V}}
\newcommand{\labeling}{\ensuremath{\ell}}
\newcommand{\dependencies}{\ensuremath{\delta}}

\newcommand{\scheck}{\text{\ding{51}}}
\newcommand{\serror}{\text{\ding{55}}}

\newcommand{\cstep}{\ensuremath{c}}

\newcommand{\para}[1]{\noindent {\bf #1. }}
\newcommand{\eg}{{\em e.g.~\xspace}}

\begin{document}

\title{Temporal Stream Logic: \\[0.2em] Synthesis beyond the
  Bools\thanks{Supported by the European Research Council (ERC) Grant
    OSARES (No.\ 683300), the Collaborative Research Center (TRR 248,
    389792660), and the National Science Foundation (NSF) Grant
    CCF-1302327.}}

\author{
       Bernd Finkbeiner\inst{1}
  \and Felix Klein\inst{1}
  \and Ruzica Piskac\inst{2}
  \and Mark Santolucito\inst{2}
}

\institute{Saarland University, Saarbrücken, Germany
  \and
Yale University, New Haven, USA}

\maketitle

\begin{abstract}
  Reactive systems that operate in environments with complex data,
  such as mobile apps or embedded controllers with many sensors, are
  difficult to synthesize.  Synthesis tools usually fail for such
  systems because the state space resulting from the discretization
  of the data is too large.  We introduce TSL, a new temporal logic
  that separates control and data. We provide a CEGAR-based synthesis
  approach for the construction of implementations that are guaranteed
  to satisfy a TSL specification for all possible instantiations of
  the data processing functions.  TSL provides an attractive trade-off
  for synthesis. On the one hand, synthesis from TSL, unlike synthesis
  from standard temporal logics, is undecidable in general. On the
  other hand, however, synthesis from TSL is scalable, because it is
  independent of the complexity of the \mbox{handled} data.  Among
  other benchmarks, we have successfully synthesized a music player
  Android app and a controller for an autonomous vehicle in the Open
  Race Car Simulator (TORCS).
\end{abstract}

\section{Introduction}
\label{sec:intro}
In reactive synthesis, we automatically translate a formal specification, typically
given in a temporal logic, into a controller that is guaranteed to
satisfy the specification. Over the past two decades there has been
much progress on reactive synthesis, both in terms of algorithms,
notably with techniques like GR(1)-synthesis~\cite{bloem2012synthesis} and bounded synthesis~\cite{Schewe:2013},
and in terms of tools, as showcased, for example, in the annual
{\sc syntcomp} competition~\cite{SYNTCOMP2017}.

In practice however, reactive synthesis has seen limited success.
One of the largest published success stories~\cite{Khalimov2014ParameterizedSynthesisCaseStudyAMBA} is the synthesis of
the AMBA bus
protocol. To push synthesis even further, automatically synthesizing a controller for an autonomous system has been recognized to be of critical importance~\cite{wongpiromsarn2013synthesis}.
Despite
many years of experience with synthesis tools, our own attempts to
synthesize such controllers with existing tools have been unsuccessful. The reason is that the tools are unable to
handle the data complexity of the controllers. The controller only
needs to switch between a small number of behaviors, like steering
during a bend, or shifting gears on high rpm. The number of control
states in a typical controller (cf. \cite{SCAV2017}) is thus not much
different from the arbiter in the AMBA case study. However, in order to correctly initiate transitions between \mbox{control states}, the driving
controller must continuously process data from more than 20 sensors.

If this data is included (even as a rough discretization) in the state
space of the controller, then the synthesis problem is much too large
to be handled by any available tools. It seems clear then, that a scalable synthesis approach must separate
control and data. If we assume that the data processing is handled by
some other approach (such as deductive synthesis~\cite{Manna:1980:DAP:357084.357090} or manual programming), is it
then possible to solve the remaining reactive synthesis problem?

In this paper, we show scalable reactive synthesis is indeed possible. Separating
data and control has allowed us to synthesize reactive systems, including an
autonomous driving controller and a music player app, that
had been impossible to synthesize with previously available tools. However, the
separation of data and control implies some fundamental changes to
reactive synthesis, which we describe in the rest of the paper. The
changes also imply that the reactive synthesis problem is no
longer, in general, decidable. We thus trade theoretical
decidability for practical scalability, which is, at least with regard to the goal of synthesizing realistic systems, definitely an attractive trade-off.

We introduce Temporal Stream Logic (\TSL), a new temporal logic that includes
\emph{updates}, such as $\upd{\name{y}}{\name{f}~\name{x}}$, and
predicates over arbitrary function terms. The update $\upd{\name{y}}{\name{f}~\name{x}}$ indicates
that the result of applying function~$ \name{f} $ to variable~$ \name{x} $ is assigned
to $\name{y}$. The implementation of predicates and functions is
not part of the synthesis problem. Instead, we look for a
system that satisfies the \TSL specification \emph{for all
possible interpretations of the functions and predicates}.

This implicit quantification over all possible interpretations provides a useful
abstraction: it allows us to \emph{independently} implement the data
processing part. On the other hand, this quantification is also the
reason for the undecidability of the synthesis problem. If a predicate
is applied to the same term \emph{twice}, it must (independently of the interpretation)
return the \emph{same} truth value. The synthesis must then implicitly
maintain a (potentially infinite) set of terms to which the predicate
has previously been applied. As we show later, this set of terms can
be used to encode PCP~\cite{post1946} for a proof of undecidability.

We present a practical synthesis approach for \TSL specifications, which is based
on bounded synthesis~\cite{Schewe:2013} and counterexample-guided abstraction refinement (CEGAR)~\cite{DBLP:journals/jacm/ClarkeGJLV03}.
We use bounded synthesis to search for an implementation up to a (iteratively growing)
bound on the number of states. This approach underapproximates the actual \TSL synthesis
problem by leaving the interpretation of the predicates to the environment.
The underapproximation allows for inconsistent behaviors: the environment might assign different
truth values to the same predicate when evaluated at different points in time,
even if the predicate is applied to the same term.
However, if we find an implementation in this underapproximation, then the CEGAR loop terminates and we have a correct implementation for the original \TSL specification. If we do not find an implementation in the underapproximation,
we compute a counter strategy for the environment. Because bounded synthesis reduces the synthesis problem to a safety game, the counter strategy is a reachability strategy that can be represented as a finite tree. We check whether the counter strategy is spurious by searching for a pair of positions in the strategy where some predicate results in different truth values when applied to the same term.
If the counter strategy is not spurious, then no implementation exists for the considered bound, and we increase the bound.
If the counter strategy is spurious, then we introduce a constraint
into the specification that eliminates the incorrect interpretation
of the predicate, and continue with the refined specification.

\tikzstyle{block} = [rectangle, draw, fill=gray!20,
    text width=7em, text centered, rounded corners, minimum height=3em, node distance=2cm]
\tikzstyle{block3} = [rectangle, draw, fill=gray!20,
    text width=7em, text centered, rounded corners, minimum height=3em, node distance=1.5cm]
\tikzstyle{block2} = [rectangle, draw,
    text width=7em, text centered, rounded corners, minimum height=3em, node distance=2cm]
\tikzstyle{input} = [rectangle, draw, fill=green!20,
    text width=7em, text centered, rounded corners, minimum height=2em, node distance=2cm]
\tikzstyle{line} = [draw, -latex']
\tikzstyle{line} = [draw, -latex']
\tikzstyle{line2} = [draw, dotted, -latex']

\tikzstyle{main} = [rectangle, draw, fill=green!30, text width=3.5em, inner sep=0.5em, text centered, rounded corners=4, minimum height=2.5em]

\tikzstyle{main2} = [rectangle, draw, fill=green!40!blue!15!, text width=2.5em, inner sep=0.5em, text centered, minimum height=2em]

\tikzstyle{main3} = [rectangle, draw, fill=blue!10!green!20, minimum width=1.4em, inner sep=0.5em, text centered, rounded corners=4, minimum height=1.4em]

\tikzstyle{conv} = [rectangle, draw, fill=blue!15, inner sep=0.5em, text centered,
 minimum height=2em,,anchor=center, drop shadow]
\tikzstyle{conv2} = [rectangle, draw, fill=white!15, inner sep=0.5em, text centered,
 minimum height=2em,,anchor=center]

\begin{figure}[t]
  \centering
  \begin{tikzpicture}[scale=1.1]
    \clip (-1.1,-4.28) rectangle (10.4,2.56);

    \node[opacity=0] at (10.7,0) {x};

    \node[main,fill=yellow!15] at (0.2,0.55) (tsl) {\large TSL};

    \node[main] at (0.2,-0.55) (tam) {\large CFM};

    \node[conv,anchor=west,minimum height=5.5em,inner sep=0pt,minimum width=13.5em] at (2.1,0) (synt) {
      \begin{tikzpicture}[>=stealth, anchor=center,minimum height=1em,minimum width=1em]
        \node at (0,0) {Synthesis};

        \node[main2] at (-0.8,-1) (ltl) {LTL};
        \node[main2,text width=3em] at (0.8,-1) (mealy) {Circuit};

        \draw[thick] (-2,-0.35) -- (-1,-0.35);
        \draw[thick] (-1,-0.35) edge[dashed] (1,-0.35);
        \draw[thick] (1,-0.35) edge[->] (2,-0.35);
        \draw (-1.7,-0.35) -- (-1.7,-1);
        \draw (-1.7,-1) edge[->] (ltl.west);
        \draw (ltl.east) edge[->] (mealy.west);
        \draw (mealy.east) -- (1.7,-1);
        \draw (1.7,-1) edge[->] (1.7,-0.35);
      \end{tikzpicture}
    };

    \node[draw,circle,fill=yellow!15] at ($ (synt.north) + (-1.9,0.5) $) (bound) {\large $ n $};

    \node[conv,minimum height=2.5em,minimum width=13.5em] at ($ (synt) + (0,-1.6) $) (trans) {
      \begin{tabular}{c}
        FRP Translator
      \end{tabular}
    };

    \node[conv,minimum height=2.5em,minimum width=13.5em] at ($ (trans) + (0,-1.1) $) (inte) {
        Project Context
        \qquad \qquad \
    };

    \node[conv,anchor=west,minimum height=2.5em] at (2.1,-3.8) (compiler) {
      \begin{tabular}{c}
        ~Compiler~ \\[-0.2em]
      \end{tabular}
    };

    \node[conv2,anchor=north,dashed] at ($ (synt.north east) + (2,0) $) (syntt) {
      \begin{tabular}{c}
        LTL \\ Synthesis Tool \\[-0.2em]
      \end{tabular}
    };

    \node at (0,0.55) (syntoutr) {};
    \node at (0,-0.4) (syntoutl) {};
    \node at (0,-2.65) (transout) {};
    \node at (0,-4.8) (compout) {};
    \node at (0,-4.7) (compin) {};

    \node[main,fill=orange!20,text width=4em] at ($ (synt.north) + (0.15,1.1) $) (counter) {
      \begin{tabular}{c}
        Counter \\
        Strategy
      \end{tabular}
    };

    \node[conv,minimum height=2.5em,minimum width=7.5em,anchor=center] at
    (tsl |- counter) (refi) {Refinement};

    \node[main,fill=red!20,text width=6em,xshift=20]  at (counter -| syntt) (unre) {\textbf{unrealizable}};

    \node[conv2,anchor=east,dashed,yshift=21,xshift=-8] at (unre.east |- trans) (pattern) {
      \begin{tabular}{c}
        Design Pattern: \\ Arrow | Applicative \\[-0.2em]
      \end{tabular}
    };

    \node[main] at (tsl |- inte) (frp) {\large FRP};

    \node[main,anchor=east] at (compiler -| inte.east) (exe) {\large EXE};

    \node[conv2,anchor=south east,dashed] at (compiler.south -| unre.east) (terms) {
      \begin{tabular}{c}
        Function \& Predicate\\ Implementations
      \end{tabular}
    };

    \node[conv2,dashed] at ($ (terms -| unre) + (0,1.7) $) (lib) {
      \begin{tabular}{c}
        FRP Library
      \end{tabular}
    };

    \node[main3] at ($ (inte) + (1.3,0) $) (module) {};

    \draw[->,>=stealth, very thick] ($ (tsl) + (-1.4,0) $) -- (tsl);
    \draw[line width=0.5em,color=black!40,->] (tsl) -- (tsl.east -| synt.west);
    \draw[line width=0.5em,color=black!40,->] (bound) -| ($ (synt.north) + (-1.15,0) $);
    \draw[line width=0.5em,color=black!40,->] (tam -| synt.west) to node[above,xshift=5] {\large \color{green!50!black} \scheck} (tam);
    \draw[line width=0.5em,color=black!40,->] (counter) to node[above,xshift=-5] {\color{black} non-spurious} (unre);
    \draw[line width=0.5em,color=black!40,->] (counter) to node[above,xshift=5] {\color{black} spurious} (refi);
    \draw[line width=0.5em,color=black!40,->] (refi.south -| tsl) -- (tsl);
    \draw[line width=0.5em,color=black!40,->] (frp -| inte.west) -- (frp);
    \draw[line width=0.5em,color=black!40,->] (tam) |- (trans);
    \draw[line width=0.5em,color=black!40,->] (frp) |- (compiler.west);

    \draw[line width=0.5em,color=black!40,->] (synt.north -| counter) to node[right,xshift=5,yshift=-1] {\large \color{red!60!black} \serror} (counter);

    \draw[line width=0.5em,color=black!40,->] (compiler.east |- exe) -- (exe);

    \draw[line width=0.5em,color=black!40,->] ($ (trans.south) + (1.3,0) $) -- (module.north);

    \draw[->,>=stealth, very thick] (syntt) -- (syntt -| synt.east);
    \draw[->,>=stealth, very thick] ($ (pattern.south west) + (0.6,0) $) |- ($ (trans.east) + (0,-0.1) $);
    \draw[->,>=stealth, very thick] (lib.south) |- ($ (inte.east) + (0,0.125) $);
    \draw[->,>=stealth, very thick] (terms.north) |- ($ (inte.east) + (0,-0.125) $);

  \end{tikzpicture}
  \caption{The TSL synthesis procedure uses a modular design. Each
    step takes input from the previous step as well as interchangeable
    modules (dashed boxes).}
  \label{fig:system}
\end{figure}
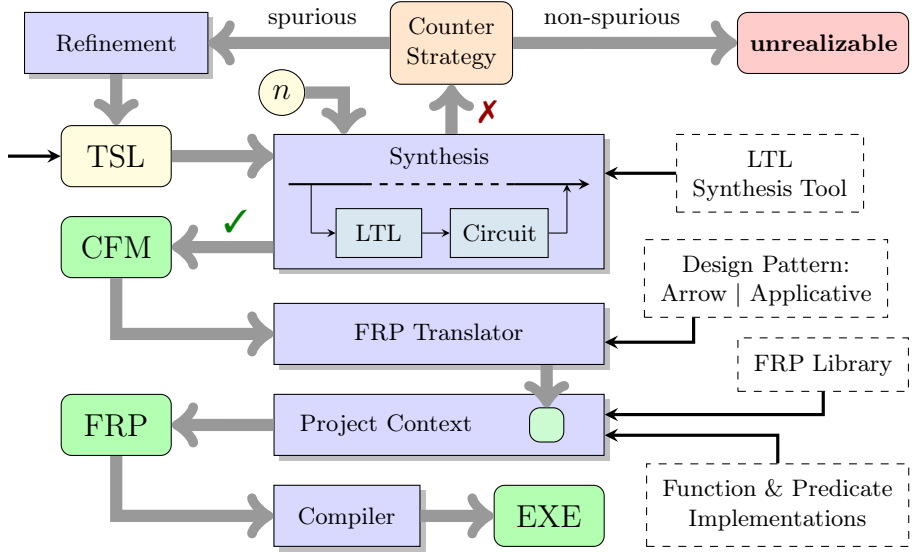

A general overview of this procedure is shown in \cref{fig:system}.
The top half of the figure depicts the bounded search for an
implementation that realizes a TSL specification using the CEGAR loop to
refine the specification. If the specification is realizable, we
proceed in the bottom half of the process, where a synthesized
implementation is converted to a control flow model (CFM) determining
the control of the system. We then specialize the CFM to Functional
Reactive Programming (FRP), which is a popular and expressive
programming paradigm for building reactive programs using functional
programming languages~\cite{hudakFRAN}. Our framework supports any FRP
library using the \textit{Arrow} or \textit{Applicative} design patterns, which covers
most of the existing FRP libraries
(e.g.~\cite{reactivebanana,clash2015,courtney2003yampa,perez2016yampa}).
Finally, the synthesized control flow is embedded into a project
context, where it is equipped with function and predicate
implementations and then compiled to an executable program.

Our experience with synthesizing systems based on \TSL specifications has been extremely positive. The synthesis works
for a broad range of benchmarks, ranging from classic reactive synthesis
problems (like escalator control), through programming exercises
from functional reactive programming, to novel
case studies like our music player app and the autonomous driving
controller for a vehicle in the Open Race Car Simulator (TORCS).

\section{Motivating Example}
\label{sec:motiv}
\newcommand{\applink}{\url{https://play.google.com/store/apps/details?id=com.mark.myapplication}.}

To demonstrate the utility of our method, we synthesized a music player Android app\footnote{\applink} from a \TSL specification.
A major challenge in developing Android apps is the temporal behavior of an app through the \textit{Android lifecycle}~\cite{Shan16}.
The Android lifecycle describes how an app should handle being paused, when moved to the background, coming back into focus, or being terminated.
In particular, \textit{resume and restart errors} are commonplace and difficult to detect and correct~\cite{Shan16}.
Our music player app demonstrates a situation in which a resume and restart error could be unwittingly introduced when programming by hand, but is avoided by providing a specification.
We only highlight the key parts of this example here to give an intuition of \TSL, leaving a more in-depth exposition to \cref{apx:musicspec}.

Our music player app utilizes the Android music player library~(\name{MP}), as well as its control interface~(\name{Ctrl}). It pauses any playing music when moved to the background (for instance if a call is received), and continues playing the currently selected track~(\name{Tr}) at the last track position when the app is resumed.
In the Android system~(\name{Sys}), the \texttt{leaveApp} method is called whenever the app moves to the background, while the \texttt{resumeApp} method is called when the app is brought back to the foreground. To avoid confusion between pausing music and pausing the app, we use \texttt{leaveApp} and \texttt{resumeApp} in place of the Android methods onPause and onResume.
A programmer might manually write code for this as shown on the left in \cref{fig:smallcode}.

The behavior of this can be directly described in \TSL as shown on the right in \cref{fig:smallcode}.
Even eliding a formal introduction of the notation for now, the specification closely matches the textual specification.
First, when the user leaves the app and the music is playing, the music pauses.
Likewise for the second part, when the user resumes the app, the music starts playing again.

\begin{figure}[t]
\vspace{-1em}
\begin{minipage}{.42\textwidth}
  \vspace{-0.8em}
\begin{lstlisting}
Sys.leaveApp()
  if (MP.musicPlaying())
    Ctrl.pause();
\end{lstlisting}
\vspace{-1em}
\begin{lstlisting}
Sys.resumeApp() {
  pos = MP.trackPos();
  Ctrl.play(Tr,pos);
}
\end{lstlisting}
\vspace{-0.8em}
\end{minipage}%
\vrule{}%
\begin{minipage}{.59\textwidth}
\vspace{-0.8em}
\begin{align*}
& \name{ALWAYS} \; \Big(\name{leaveApp} \ \, \name{Sys} \; \wedge \; \name{musicPlaying} \ \, \name{MP} \\[-0.5em]
& \quad \hspace{2.5em} \impl \upd{\name{Ctrl}}{\const{pause}} \Big) \\[0.8em]
& \name{ALWAYS} \; \Big(\name{resumeApp} \ \, \name{Sys}  \\[-0.5em]
& \quad \hspace{2.5em} \impl  \upd{\name{Ctrl}}{\name{play} \ \, \name{Tr} \ \, (\name{trackPos} \ \, \name{MP})} \Big)
\end{align*}
\vspace{-0.8em}
\end{minipage}
\vspace{-0.5em}
\caption{Sample code and specification for the music player app.}
\label{fig:smallcode}
\end{figure}
However, assume we want to change the behavior so that the music only plays on resume when the music had been playing before leaving the app in the first place.
In the manually written program, this new functionality requires an additional variable~\texttt{wasPlaying} to keep track of the music state.
Managing the state requires multiple changes in the code as shown on the left in \cref{fig:bigcode}.
The required code changes include: a conditional in the \texttt{resumeApp} method, setting \texttt{wasPlaying} appropriately in two places in \texttt{leaveApp}, and providing an initial value.
Although a small example, it demonstrates how a minor change in functionality may require wide-reaching code changes.
In addition, this change introduces a globally scoped variable, which then might accidentally be set or read elsewhere.
In contrast, it is a simple matter to change the TSL specification to reflect this new functionality.
Here, we only update one part of the specification to say that if the user leaves the app and the music is playing, the music has to play again as soon as the app resumes.

\begin{figure}[t]
\begin{minipage}{.44\textwidth}
\vspace{-0.8em}
\begin{lstlisting}
bool wasPlaying = false;
\end{lstlisting}
\vspace{-0.5em}
\begin{lstlisting}
Sys.leaveApp()
  if (MP.musicPlaying()) {
    wasPlaying = true;
    Ctrl.pause();
  }
  else
    wasPlaying = false;
\end{lstlisting}
\vspace{-0.5em}
\begin{lstlisting}
Sys.resumeApp()
  if (wasPlaying) {
    pos = MP.trackPos();
    Ctrl.play(Tr,pos);
  }
\end{lstlisting}
\vspace{-0.8em}
\end{minipage}%
\vrule{}%
\begin{minipage}{.57\textwidth}
\begin{align*}
& \,\name{ALWAYS} \; \Big( (\name{leaveApp} \ \, \name{Sys} \; \wedge \ \name{musicPlaying} \ \, \name{MP} \\
& \,\quad \hspace{2.5em} \impl \upd{\name{Ctrl}}{\const{pause}} ) \\[0.5em]
& \,\quad \hspace{2.5em} \; \wedge \, (\upd{\name{Ctrl}}{\name{play} \ \, \name{Tr} \ (\name{trackPos} \ \, \name{MP})} \ \\
& \,\quad \hspace{2.5em} \phantom{\impl} \ \ \name{AS\_SOON\_AS} \ \ \name{resumeApp} \ \, \name{Sys} ) \Big)
\end{align*}
\end{minipage}
\vspace{-0.5em}
\caption{The effect of a minor change in functionality on code versus a specification.}
\label{fig:bigcode}
\end{figure}
Synthesis allows us to specify a temporal behavior without worrying about the implementation details.
In this example, writing the specification in \TSL has eliminated the need of an additional state variable, similarly to a higher order \texttt{map} eliminating the need for an iteration variable.
However, in more complex examples the benefits compound, as \TSL provides a modular interface to specify behaviors, offloading the management of multiple interconnected temporal behaviors from the user to the synthesis engine.

\section{Preliminaries}
\label{sec:prelim}
We assume time to be discrete and denote it by the set $ \dtime $ of
positive integers.  A value is an arbitrary object of arbitrary
type. $ \values $ denotes the set of all values.  The Boolean values
are denoted by $ \bool \subseteq \values $.  A
stream~$ \signal \from \dtime \to \values $ is a function fixing
values at each point in time.
An $ n $-ary function~$ f \from \values^{n} \to \values $ determines
new values from $ n $ given values, where the set of all functions (of
arbitrary arity) is given by~$ \functions $. Constants are functions
of arity 0. Every constant is a value, i.e., is an element of
$ \functions \cap \values $. An $ n $-ary
predicate~$ p \from \values^{n} \to \bool $ checks a property
over~$ n $ values.  The set of all predicates (of arbitrary arity) is
given by~$ \predicates $, where $ \predicates \subseteq \functions $.
We use $ \fspace{\hspace{-1pt}A}{B} $ to denote the set of all total functions with
domain~$ A $ and image~$ B $.

In the classical synthesis setting, inputs and outputs are vectors of
Booleans, where the standard abstraction treats inputs and outputs as
atomic propositions $ \mathcal{I} \cup \mathcal{O} $, while their
Boolean combinations form an
alphabet~\mbox{$ \Sigma = 2^{\mathcal{I} \cup \mathcal{O}}
  $}. Behavior then is described through infinite sequences
$ \alpha = \alpha(0)\alpha(1)\alpha(2) \ldots \in \Sigma^{\omega} $.
A \textit{specification} describes a relation between input
sequences~\mbox{$ \alpha \in (2^{\mathcal{I}})^{\omega} $} and output
sequences~\mbox{$ \beta \in (2^{\mathcal{O}})^{\omega} $}. Usually,
this relation is not given by explicit sequences, but by a fomula in a
temporal logic.  The most popular such logic is Linear Temporal Logic
(LTL)~\cite{Pnueli:1977}, which uses Boolean connectives to specify
behavior at specific points in time, and temporal connectives, to
relate sub-specifications over time. The realizability and synthesis
problems for LTL are 2\textsc{ExpTime}-complete~\cite{PnueliR89}.

An implementation describes a realizing strategy, formalized via infinite trees. A $ \Phi $-labeled
and \mbox{$ \Upsilon $-}branching tree is a
function~$ \sigma \from \Upsilon^{*} \to \Phi $, where $ \Upsilon $
denotes the set of branching directions along a tree. Every node of
the tree is given by a finite prefix $ v \in \Upsilon^{*} $, which
fixes the path to reach a node from the root. Every node is labeled
by an element of $ \Phi $. For infinite
paths~$ \nu \in \Upsilon^{\omega} $, the
branch~$ \branch{\sigma}{\nu} $ denotes the sequence of labels that appear
on $ \nu $, i.e.,
$ \forall t \in \nats. \ (\branch{\sigma}{\nu})(t) =
\sigma(\nu(0) \ldots \nu(t-1)) $.

\section{Temporal Stream Logic}
\label{sec:TSL}
We present a new logic: Temporal Stream Logic (\TSL), which is
especially designed for synthesis and allows for the manipulatation of infinite
streams of arbitrary (even non-enumerative, or higher order) type. It
provides a straightforward notation to specify how outputs are
computed from inputs, while using an intuitive interface to access
time. The main focus of \TSL is to describe temporal control
flow, while abstracting away concrete implementation details. This not
only keeps the logic intuitive and simple, but also allows a user to identify
problems in the control flow even without a concrete implementation at
hand. In this way, the use of \TSL scales up to any required abstraction, such as API
calls or complex algorithmic transformations.

\medskip

\noindent \textit{Architecture} A TSL formula~$ \varphi $ specifies a
reactive system that in every time step processes a finite number of inputs~$ \inames $
and produces a finite number of outputs~$ \onames $. Furthermore, it
uses cells~$ \cells $ to store a value computed at time~$ t $,
which can then be reused in the next time step~$ t + 1 $. An overview
of the architecture of such a system is given in \cref{fig:tslarchitecture}. In terms
of behavior, the environment produces infinite streams of input data,
while the system uses pure (side-effect free) functions to transform
the values of these input streams in every time step. After their
transformation, the data values are either passed to an output stream
or are passed to a cell, which pipes the output value from one time step back to the corresponding input value of the next.
The behaviour of the system is captured by its infinite execution over time.

\begin{figure}[t]
  \centering

  \begin{subfigure}[b]{0.58\textwidth}
  \begin{tikzpicture}[scale=0.8]

    \node[anchor=east,inner sep=0pt] at (-2.8,-1.5) {
      \small
      \begin{tabular}{c}
        inputs: \\[0.1em] $ \inames $
      \end{tabular}
    };

    \node at (0,1.38) {
      \small
      cells: $ \cells $
    };

    \node[anchor=west,inner sep=0pt] at (2.8,-1.5) {
      \small
      \begin{tabular}{c}
        outputs: \\[0.1em]
        $ \onames $
      \end{tabular}
    };

    \node at (0,0) {
      \begin{tikzpicture}[xscale=0.8,yscale=0.56]
        \node[fill, fill=blue!30,minimum height=5.5em, minimum width=10.5em] (C) {};

        \node at (C) {
          \small
          \begin{tabular}{c}
            \textit{reactive system} \\[0.4em]
            \textit{implementing a} \\[0.4em]
            \textit{TSL specification~$ \varphi $}
            \end{tabular}
        };

        \node[minimum size=0.9em] (H0) at (0,1.95) {};
        \node[minimum size=0.9em] (H1) at (0,2.6) {};
        \node[minimum size=0.9em] (H2) at (0,3.8) {};

        \path[->,>=stealth,line width=0.7pt]
        ($ (C.west) + (-0.6,-1.3) $) edge ($ (C.west) + (0,-1.3) $)
        ($ (C.west) + (-0.6,-0.6) $) edge ($ (C.west) + (0,-0.6) $)
        ($ (C.west) + (-0.6,-0.3) $) edge ($ (C.west) + (0,-0.3) $)
        ($ (C.east) + (0,-1.3) $) edge ($ (C.east) + (0.6,-1.3) $)
        ($ (C.east) + (0,-0.6) $) edge ($ (C.east) + (0.6,-0.6) $)
        ($ (C.east) + (0,-0.3) $) edge ($ (C.east) + (0.6,-0.3) $)
        ;

        \node at ($ (C.west) + (-0.35,-0.85) $) {\scalebox{0.6}{$ \vdots $}};
        \node at ($ (C.east) + (0.25,-0.85) $) {\scalebox{0.6}{$ \vdots $}};

        \draw[line width=0.7pt,-,>=stealth,gray]
        ($ (C.east) + (0,1.3) $) -- (2.7,1.3) |- (H0);
        \draw[line width=0.7pt,->,>=stealth,gray]
        (H0) -| (-2.7,1.3) -- ($ (C.west) + (0,1.3) $);

        \draw[line width=0.7pt,-,>=stealth,gray]
        ($ (C.east) + (0,1) $) -- (3,1) |- (H1);
        \draw[line width=0.7pt,->,>=stealth,gray]
        (H1) -| (-3,1) |- ($ (C.west) + (0,1) $);

        \draw[line width=0.7pt,-,>=stealth,gray]
        ($ (C.east) + (0,0.3) $) -- (3.6,0.3) |- (H2);
        \draw[line width=0.7pt,->,>=stealth,gray]
        (H2) -| (-3.6,0.3) |- ($ (C.west) + (0,0.3) $);

        \node at ($ (C.west) + (-0.35,0.75) $) {\scalebox{0.6}{$ \vdots $}};
        \node at ($ (C.east) + (0.25,0.75) $) {\scalebox{0.6}{$ \vdots $}};

        \fill[fill=orange!60]
        ($ (H0.north west) + (0,-0.1) $) --
        ($ (H0.south west) + (0,0.1) $) --
        ($ (H0.south west) + (0.1,0) $) --
        ($ (H0.south east) + (-0.1,0) $) --
        ($ (H0.south east) + (0,0.1) $) --
        ($ (H0.north east) + (0,-0.1) $) --
        ($ (H0.north east) + (-0.1,0) $) --
        ($ (H0.north west) + (0.1,0) $) --
        cycle;

        \fill[fill=orange!60]
        ($ (H1.north west) + (0,-0.1) $) --
        ($ (H1.south west) + (0,0.1) $) --
        ($ (H1.south west) + (0.1,0) $) --
        ($ (H1.south east) + (-0.1,0) $) --
        ($ (H1.south east) + (0,0.1) $) --
        ($ (H1.north east) + (0,-0.1) $) --
        ($ (H1.north east) + (-0.1,0) $) --
        ($ (H1.north west) + (0.1,0) $) --
        cycle;

        \fill[fill=orange!60]
        ($ (H2.north west) + (0,-0.1) $) --
        ($ (H2.south west) + (0,0.1) $) --
        ($ (H2.south west) + (0.1,0) $) --
        ($ (H2.south east) + (-0.1,0) $) --
        ($ (H2.south east) + (0,0.1) $) --
        ($ (H2.north east) + (0,-0.1) $) --
        ($ (H2.north east) + (-0.1,0) $) --
        ($ (H2.north west) + (0.1,0) $) --
        cycle;

      \end{tikzpicture}
    };
  \end{tikzpicture}
\caption{Architecture \mbox{\ }}
\label{fig:tslarchitecture}
\end{subfigure}
\begin{subfigure}[b]{0.41\textwidth}
  \begin{tikzpicture}

    \fill[gray!20,rounded corners=3] (-2.5,-2.5) rectangle (2.5,0.7);

    \node[anchor=center] at (0,0.45) {
      \textbf{Function Term:}
    };

    \node[anchor=center] at (0,0) {
      $ \fterm \ := \ \; \name{s}_{\name{i}} \! \! \sep \! \name{f} \ \, \fterm^{0} \
      \, \fterm^{1} \ \, \cdots \ \, \fterm^{n-1} $
    };

    \draw (-2.5,-0.35) -- (2.5,-0.35);

    \node[anchor=center] at (0,-0.65) {
      \textbf{Predicate Term:}
    };

    \node[anchor=center] at (0,-1.1) {
      $ \pterm \ := \ \;
      \name{p} \ \, \fterm^{0} \ \, \fterm^{1} \ \, \cdots \ \, \fterm^{n-1}  $
    };

    \draw (-2.5,-1.45) -- (2.5,-1.45);

    \node[anchor=center] at (0,-1.8) {
      \textbf{Update:}
    };

    \node[anchor=center] at (0,-2.25) {
      $ \upd{\name{s}_{\name{o}}}{\fterm} $
    };
  \end{tikzpicture}
  \caption{Term Definitions \mbox{\ }}
  \label{fig:termdefinitions}
\end{subfigure}

\caption{General architecture of reactive systems that are specified in TSL
  on the left, and the structure of function, predicate and updates
  on the right.}
\end{figure}
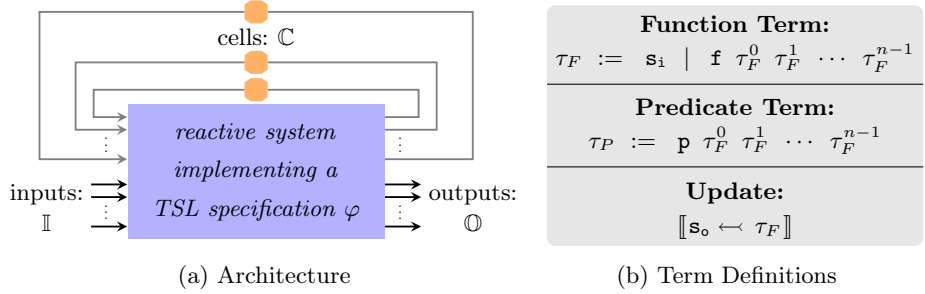

\medskip

\fussy

\noindent \textit{Function Terms, Predicate Terms, and Updates} In \TSL we
differentiate between two elements: we use purely functional
transformations, reflected by \mbox{functions~$ f \in \functions $}
and their compositions, and \mbox{predicates~$ p \in \predicates $},
used to control how data flows inside the system.
To argue about both
elements we use a term based notation, where we distinguish between
function terms~$ \fterm $ and predicate terms~$ \pterm $,
respectively. Function terms are either constructed from inputs or
cells \mbox{($ \name{s}_{\name{i}} \in \inames \cup \cells $)}, or from
functions, recursively applied to a set of function terms. Predicate
terms are constructed similarly, by applying a predicate to a set of
function terms.
\sloppy
Finally, an update takes the result of a function computation and passes it either to an output
or a cell ($ \name{s}_{\name{o}} \in \onames \cup \cells $). An overview of the syntax
of the different term notations is given in \cref{fig:termdefinitions}.
Note that we use curried argument notation similar to functional
programming languages.

We denote sets of function and predicate terms, and updates by
$ \fterms $, $ \pterms $ and $ \uterms $, respectively, where
$ \pterms \subseteq \fterms $.  We use $ \fnames $ to denote the set
of function literals and $ \pnames \subseteq \fnames $ to denote the
set of predicate literals, where the literals $ \name{s}_{\name{i}} $,
$ \name{s}_{\name{o}} $, $ \name{f} $ \linebreak and~$ \name{p} $ are
symbolic representations of inputs and cells, outputs and cells, and
functions and predicates, respectively.  Literals are used to
construct terms as shown in \cref{fig:termdefinitions}. Since we use a
symbolic representation, functions and predicates are not tied to a
specific implementation. However, we still classify them according to
their arity, i.e., the number of function terms they are applied to,
as well as by their type: input, output, cell, function or
predicate. Furthermore, terms can be compared syntactically using the
equivalence relation~$ \equiv $.  To assign a semantic interpretation
to functions, we use an assignment function
\mbox{$ \assign{\cdot} \from \fnames \to \functions $}.

\medskip

\noindent \textit{Inputs, Outputs, and Computations} We consider momentary
inputs \mbox{$ i \in \fspace{\inames}{\values} $}, which are
assignments of inputs~$ \name{i} \in \inames $ to values
$ v \in \values $. For the sake of readability let
$ \inputs = \fspace{\inames}{\values} $. Input streams are infinite
sequences~\mbox{$ \iota \in \inputs^{\hspace{0.2pt}\omega} $}
consisting of infinitely many momentary inputs.

Similarly, a momentary output~$ o \in \fspace{\onames}{\values} $ is
an assignment of outputs~$ \name{o} \in \onames $ to values
$ v \in \values $, where we also use
$ \outputs = \fspace{\onames}{\values} $. Output streams are infinite
sequences~$ \varrho \in \outputs^{\hspace{0.5pt}\omega} $. To capture
the behavior of a cell, we introduce the notion of a
computation~$ \comp $.
A computation fixes the function terms that are used to compute outputs and cell updates,
without fixing semantics of function literals.
 Intuitively, a computation only determines which
function terms are used to compute an output, but abstracts from actually
computing it.

The basic element of a computation is a computation
step~$ \cstep \in \fspace{\onames \cup \cells}{\fterms} $, which is an
assignment of outputs and
cells~$ \name{s}_{\name{o}} \in \onames \cup \cells $ to function
terms~$ \fterm \in \fterms $. For the sake of readability let
$ \comps = \fspace{\onames \cup \cells}{\fterms} $. A computation step
fixes the control flow behaviour at a single point in time. A
computation~\mbox{$ \comp \in \comps^{\omega} $} is an infinite sequence of
computation steps.

As soon as input streams, and function and predicate implementations
are known, computations can be turned into output streams. To this
end, let $ \assign{\cdot} \from \fnames \to \functions $ be some
function assignment.  Furthermore, assume that there are predefined
constants~$ \inits \in \functions \cap \values $ for every
cell~$ \name{c} \in \cells $, which provide an initial
value for each stream at the initial point in time. To receive an
output stream from a computation~$ \comp \in \comps^{\omega} $ under
the input stream $ \iota $, we use an evaluation
function~$ \eval \from \comps^{\omega} \times
\inputs^{\hspace{0.2pt}\omega} \times \dtime \times \fterms \to
\values $:
\begin{eqnarray*}
  \eval(\comp, \iota, t, \name{s}_{\name{i}}) & = &
    \begin{cases}
      \iota(t)(\name{s}_{\name{i}}) & \text{if } \name{s}_{\name{i}} \in \inames \\
      \initsE &
      \text{if } \name{s}_{\name{i}} \in \cells \ \wedge \ t = 0 \\
      \eval(\comp, \iota, t-1, \comp(t-1)(\name{s}_{\name{i}})) &
      \text{if } \name{s}_{\name{i}} \in \cells \ \wedge \ t > 0
    \end{cases}
  \\[0.5em]
  \eval(\comp, \iota, t, \name{f} \ \term_{0} \ \cdots \ \term_{m-1}) & = &
  \assign{\name{f}} \ \eval(\comp,\iota, t,\term_{0}) \
  \cdots \ \eval(\comp,\iota, t,\term_{m-1})
\end{eqnarray*}
Then
$ \varrho_{\hspace{-1pt}\langle \hspace{-1pt}\cdot
  \hspace{-1pt}\rangle \hspace{-1pt}, \comp, \iota} \in
\outputs^{\hspace{0.5pt}\omega} $ is defined via
$ \varrho_{\hspace{-1pt}\langle \hspace{-1pt}\cdot
  \hspace{-1pt}\rangle \hspace{-1pt}, \comp, \iota}(t)(\name{o}) =
\eval(\comp, \iota, t, \name{o}) $ for all $ t \in \dtime $,
$ \name{o} \in \onames $.

\medskip
\smallskip

\noindent \textit{Syntax} Every TSL formula~$ \varphi $ is built
according to the following grammar:
\begin{equation*}
  \varphi \ \ := \ \
  \term \in \pterms\cup \uterms
  \!\sep\! \neg \varphi
  \!\sep\! \varphi \wedge \varphi
  \!\sep\! \LTLnext \varphi
  \!\sep\! \varphi \LTLuntil \varphi
\end{equation*}
An atomic proposition $\tau$ consists either of a predicate term, serving as a
Boolean interface to the inputs, or of an update, enforcing a
respective flow at the current point in time. Next, we have the
Boolean operations via negation and conjunction, that allow us to express
arbitrary Boolean combinations of predicate evaluations and
updates. Finally, we have the temporal operator next:
$ \LTLnext \psi $, to specify the behavior at the next point in time
and the temporal operator until:~$ \vartheta \LTLuntil \psi $, which
enforces a property~$ \vartheta $ to hold until the property~$ \psi $
holds, where $ \psi $ must hold at some point in the future
eventually.
\medskip
\smallskip

\noindent \textit{Semantics} Formally, this leads to the following
semantics.  Let $ \assign{\cdot} \from \fnames \to \functions $,
\mbox{$ \iota \in \inputs^{\hspace{0.2pt}\omega} $}, and
$ \comp \in \comps^{\omega} $ be given, then the validity of a $ \TSL$
formula~$ \varphi $ with respect to $ \comp $ and $ \iota $ is defined inductively
over $ t \in \dtime $ via:
\begin{equation*}
  \begin{array}{lcl}
    \\[-1.8em]
    \comp, \iota, t \sats \name{p} \ \term_{0} \ \cdots \ \term_{m-1} & \ :\Leftrightarrow \ \
    & \eval(\comp,\iota,t,\name{p} \ \term_{0} \ \cdots \ \term_{m-1}) \\[0.2em]
    \comp, \iota, t \sats \upd{\name{s}}{\!\term} & :\Leftrightarrow
    & \comp(t)(\name{s}) \equiv \term \\[0.2em]
    \comp, \iota, t \sats \neg \psi & :\Leftrightarrow
    & \comp, \iota, t \nsats \psi \\[0.2em]
    \comp, \iota, t \sats \vartheta \wedge \psi & :\Leftrightarrow
    & \comp, \iota, t \sats \vartheta \ \wedge \ \comp, \iota, t \sats \psi \\[0.2em]
    \comp, \iota, t \sats \LTLnext \psi & :\Leftrightarrow
    & \comp, \iota, t+1 \sats \psi \\[0.2em]
    \comp, \iota, t \sats \vartheta \LTLuntil \psi & :\Leftrightarrow
    & \exists t'' \geq t. \ \
         \forall t \leq t' < t''. \ \ \comp, \iota, t' \sats \vartheta \ \,
         \wedge \ \, \comp, \iota, t'' \sats \psi
  \end{array}
\end{equation*}
Consider that the satisfaction of a predicate depends on the current
computation step and the steps of the past, while for updates it only
depend on the current computation step. Furthermore, updates are only
checked syntactically, while the satisfaction of predicates depends on
the given assignment~$ \assign{\cdot} $ and the input stream
$ \iota $.
We say that $ \comp $ and $ \iota $ satisfy $ \varphi $, denoted by
$ \comp, \iota \sats \varphi$, if $ \comp, \iota, 0 \sats \varphi
$.

Beside the basic operators we have the standard derived Boolean
operators, as well as the derived temporal operators:
\textit{release}~$ \varphi \LTLrelease \psi \equiv \neg ((\neg \psi)
\LTLuntil (\neg \varphi)) $,
\textit{finally}~$ \LTLfinally \varphi \equiv \emph{true} \LTLuntil
\varphi $,
\textit{always}~$ \LTLglobally \varphi \equiv \emph{false} \LTLrelease
\varphi $, the \textit{weak} version of \textit{until}
$ \varphi \LTLweakuntil \psi \equiv (\varphi \LTLuntil \psi) \vee
(\LTLglobally \varphi) $, and \textit{as soon
  as}~$ \varphi \mathop{\mathcal{A}}\hspace{0.5pt} \psi \equiv \neg
\psi \LTLweakuntil (\psi \wedge \varphi) $.

\medskip

\noindent \textit{Realizability} We are interested in the following
realizability problem: given a $ \TSL $ formula~$ \varphi $, is there
a strategy~$ \sigma \in \fspace{\inputs^{+}}{\comps} $ such that for every
input $ \iota \in \inputs^{\omega} $ and function implementation
$ \assign{\cdot} \from \fnames \to \functions $, the branch
$ \branch{\sigma}{\iota} $ satisfies $ \varphi $, i.e.,
\begin{equation*}
  \exists \sigma \in \fspace{\inputs^{+}}{\comps}. \ \, \forall \iota \in \inputs^{\hspace{0.2pt}\omega}. \ \, \forall \assign{\cdot} \from
  \fnames \to \functions. \ \, \branch{\sigma}{\iota}, \iota \sats \varphi
\end{equation*}
If such a strategy~$ \sigma $ exists, we say $ \sigma $ realizes
$ \varphi $. If we additionally ask for a concrete instantiation of
$ \sigma $, we consider the synthesis problem of TSL.

\section{TSL Properties}
\label{sec:props}
In order to synthesize programs from TSL specifications, we give an overview of the first part of our synthesis process, as shown in \cref{fig:system}.
First we show how to approximate the semantics of TSL through a reduction to LTL.
However, due to the approximation, finding a realizable strategy immediately may fail.
Our solution is a CEGAR loop that improves the approximation.
This CEGAR loop is necessary, because the realizability problem of TSL is undecidable in general.

\medskip

\noindent \textit{Approximating TSL with LTL} We approximate TSL
formulas with weaker LTL formulas.  The approximation reinterprets the
syntactic elements, $\pterms$ and $\uterms$, as atomic propositions
for LTL. This strips away the semantic meaning of the function
application and assignment in TSL, which we reconstruct by later
adding assumptions lazily to the LTL formula.

Formally, let $ \pterms $ and $ \uterms $ be the finite sets of
predicate terms and updates, which appear in
$ \varphi_{\textit{TSL}} $, respectively. For every assigned signal, we
partition $ \uterms $ into
$ \biguplus_{\name{s}_{\name{o}} \in \onames \cup \cells}
\uterms^{\hspace{0.5pt}\name{s}_{\name{o}}} $. For every
$ \name{c} \in \cells $ let \mbox{$\utermsp^{\hspace{0.5pt}\name{c}} =
  \uterms^{\hspace{0.5pt}\name{c}} \cup \set{ \upd{\name{c}}{\name{c}}
  } $}, for $ \name{o} \in \onames $ let
$ \utermsp^{\hspace{0.5pt}\name{o}} = \uterms^{\hspace{0.5pt}\name{o}}
$, and let
$ \utermsp = \bigcup_{\name{s}_{\name{o}} \in \onames \cup \cells}
\utermsp^{\hspace{0.5pt}\name{s}_{\name{o}}} $.  We construct the LTL
formula~$ \varphi_{\textit{LTL}} $ over the input
propositions~$ \pterms $ and output propositions $ \utermsp $ as
follows:
\begin{equation*}
  \varphi_{\textit{LTL}} \, = \;
  \LTLglobally \Big ( \bigwedge_{\name{s}_{\name{o}} \in \onames \cup \cells} \,
  \bigvee_{\term \in \utermsp^{\hspace{0.5pt}\name{s}_{\name{o}}}}
  \big( \term \; \wedge \bigwedge_{\term' \in
    \utermsp^{\hspace{0.5pt}\name{s}_{\name{o}}} \setminus
    \set{ \term }} \neg \, \term' \big)  \Big) \ \wedge \
  \textsc{SyntacticConversion}\big(\varphi_{\textit{TSL}}\big)
\end{equation*}
Intuitively, the first part of the equation partially reconstructs the semantic meaning of updates by ensuring that a signal is not updated with multiple values at a time.
The second part extracts the reactive constraints of the TSL formula without the semantic meaning of functions and updates.
\begin{theorem}
  \label{thm:tsl2ltl} If $ \varphi_{\textit{LTL}} $ is realizable, then $ \varphi_{\textit{TSL}} $ is realizable.
\end{theorem}
\noindent The proof of \cref{thm:tsl2ltl} is given in \cref{proof:tsl2ltl}.
Note that unrealizability of $\varphi_{\textit{LTL}} $ does not imply that $ \varphi_{\textit{TSL}}$ is unrealizable.
It may be that we have not added sufficiently many environment assumptions to the approximation in order for the system to produce a realizing strategy.

\medskip

\label{ex:asLTL}

\begin{figure*}[t]
    \centering
    \begin{subfigure}[t]{0.28\textwidth}
      \centering
      $ \begin{array}{c}
          \\[-0.5em]
          \LTLglobally \; (\upd{\name{y}}{\name{y}} \, \vee \, \upd{\name{y}}{\name{x}}) \\[0.2em]
          \wedge \ \LTLeventually \, \name{p} \ \name{x} \, \impl \,
          \LTLeventually \, \name{p}\ \name{y} \\[-0.5em]
          \
        \end{array} $
        \caption{TSL specification}
\label{eq:tslSimple}
    \end{subfigure}%
    ~ ~
    \begin{subfigure}[t]{0.30\textwidth}
      \centering
      $ \begin{array}{c}
          \LTLglobally \; \neg ( \name{y\_to\_y} \, \wedge \, \name{x\_to\_y}) \\[0.2em]
          \wedge \ \LTLglobally \; (\name{y\_to\_y} \, \vee \, \name{x\_to\_y}) \\[0.2em]
          \wedge \ \LTLeventually \, \name{p\_x} \, \impl \
          \LTLeventually \, \name{p\_y}
        \end{array} $
\caption{initial approximation}
\label{eq:ltlSimple}
    \end{subfigure}%
    ~\;
    \begin{subfigure}[t]{0.35\textwidth}
      \centering
      \vspace{-1em}
  \begin{tikzpicture}[->,>=stealth',shorten >=1pt,auto,node distance=2.8cm,initial text=]
    \tikzstyle{every state}=[fill=blue!20,draw,text=white,minimum size=1.5em]
    \node[initial,state] (A)                    {};
    \path (A) edge [loop right] node {$\name{p\_x} \; \wedge \; \neg \, \name{p\_y}$} (A);
  \end{tikzpicture}
  \vspace{1.1em}
\caption{spurious counter-strategy}
\label{eq:tslSimpleSoln}
    \end{subfigure}
    \caption{
      A TSL specification~(a) with input~\name{x} and cell~\name{y} that is realizable. A winning strategy is to save~\name{x} to \name{y} as soon as $ \name{p}(\name{x}) $ is satisfied. However, the initial approximation~(b), that is passed to an LTL synthesis solver, is unrealizable, as proven through the counter-strategy~(c) returned by the LTL solver.}
    \label{fig:approx}
\end{figure*}

\noindent \textit{Example} As an example, we present a simple TSL specification in \cref{eq:tslSimple}.
The specification asserts that the environment provides an input~\name{x} for which the predicate~$ \name{p}~\name{x} $ will be satisfied eventually. The system must guarantee that eventually $ \name{p}~\name{y} $ holds.
According to the semantics of TSL the formula is realizable. The system can take the value of $ \name{x} $ when $ \name{p}~\name{x} $ is true and save it to $ \name{y} $, thus guaranteeing that $ \name{p}~\name{y} $ is satisfied eventually.
This is in contrast to LTL, which has no semantics for pure functions - taking the evaluation of $ \name{p}~\name{y} $ as an environmentally controlled value that does not need to obey the consistency of a pure function.

\medskip

\noindent \textit{Refining the LTL Approximation} It is possible that the LTL solver returns a counter-strategy for the environment although the original TSL specification is realizable.
We call such a counter-strategy \textit{spurious} as it exploits the additional freedom of LTL to violate the purity of predicates as made possible by the underapproximation.
Formally, a counter-strategy is an infinite tree $ \pi \from \comps^{*} \to 2^{\pterms} $, which provides predicate evaluations in response to possible update assignments of function terms~$ \fterm \in \fterms $ to outputs~$ \name{o} \in \onames $.
W.l.o.g.\ we can assume that $ \onames $, $ \fterms $ and $ \pterms $ are finite, as they can always be restricted to the outputs and terms that appear in the formula.
A counter-strategy is spurious, iff there is a branch~$ \branch{\pi}{\comp} $ for some computation~$ \comp \in \comps^{\omega} $, for which the strategy chooses an inconsistent evaluation of two equal predicate terms at different points in time, i.e.,
\begin{equation*}
  \begin{array}{l}
    \exists \comp \in \comps^{\omega}. \ \exists t, t' \in \dtime. \ \exists \pterm \in \pterms. \\[0.2em]
    \qquad \pterm \in \pi(\comp(0)\comp(1)\ldots\comp(t-1)) \, \wedge \, \pterm \notin \pi(\comp(0)\comp(1)\ldots\comp(t'-1)) \ \wedge \\[0.2em]
    \qquad \forall \assign{\cdot} \from \fnames \to \functions. \ \eval(\comp, \branch{\pi}{\comp}, t, \pterm) \, = \, \eval(\comp, \branch{\pi}{\comp}, t', \pterm).
  \end{array}
\end{equation*}
Note that a non-spurious strategy can be inconsistent along multiple
branches. Due to the definition of realizability
the environment can choose function and
predicate assignments differently against every system strategy
accordingly.

By purity of predicates in TSL the environment is forced to
always return the same value for predicate evaluations on equal
values. However, this semantic property cannot be enforced implicitly
in LTL.  To resolve this issue we use the returned counter-strategy to
identify spurious behavior in order to strengthen the LTL
underapproximation with additional environment assumptions.
After adding the derived assumptions, we re-execute the LTL synthesizer to check whether the
added assumptions are sufficient in order to obtain a winning strategy
for the system.  If the solver still returns a spurious strategy, we
continue the loop in a CEGAR fashion until the set of added
assumptions is sufficiently complete.  However, if a non-spurious strategy is
returned, we have found a proof that the given
TSL specification is indeed unrealizable and terminate.

\goodbreak

\begin{algorithm}[t]
  \small
  \caption{Check-Spuriousness} \label{alg:spurious}
  \begin{algorithmic}[1]
    \Require{bound~$ b $, counter-strategy~$ \pi \from \comps^{*}\!\! \to \! 2^{\pterms} $ (finitely represented using $ m $ states)}

    \vspace{0.3em}

    \ForAll{$ v \in \comps^{m \cdot b},\, \pterm \in \pterms, \, t,t' \in \set{ 0,1,\ldots,m\cdot b - 1} $}
      \If{$ \evalid(v,\iota_{\name{id}},t,\pterm) \equiv \evalid(v,\iota_{\name{id}},t',\pterm) \wedge \mbox{\qquad} \qquad \qquad \qquad \qquad \qquad \qquad \qquad $ $ \mbox{\ }\hspace{2.5em} \pterm \in \pi(v_{0}\ldots v_{t-1}) \wedge \pterm \notin \pi(v_{0}\ldots v_{t'-1}) $}
      \State \quad $ w \gets \texttt{reduce}\,(v,\pterm,t,t') $
      \State \quad {\textbf{return} \ $ \LTLglobally \big(\! \bigwedge_{i=0}^{t-1} \LTLnext^{i}\! w_{i} \, \wedge \, \bigwedge_{i = 0}^{t'-1} \LTLnext^{i}\! w_{i} \,\rightarrow\, (\LTLnext^{t}\! \pterm \leftrightarrow \LTLnext^{t'} \!\! \pterm) \big) $}
      \EndIf
    \EndFor
    \State {\textbf{return} \ \texttt{``non-spurious''}}
  \end{algorithmic}
  \vspace{-0.2em}
\end{algorithm}

\cref{alg:spurious} shows how a returned counter-strategy~$ \pi $ is
checked for being spurious. To this end, it is sufficient to
check~$ \pi $ against system strategies bounded by the given
bound~$ b $, as we use bounded
synthesis~\cite{Schewe:2013}. Furthermore, we can assume w.l.o.g.\
that~$ \pi $ is given by a finite state representation, which is
always possible due to the finite model guarantees of LTL. Also note
that~$ \pi $, as it is returned by the LTL synthesizer, responses to
sequences of sets of updates~$ (2^{\utermsp})^{*} $. However, in our
case $ (2^{\utermsp})^{*} $ is an alternative
representation~of~$ \comps^{*} $, due to the additional constraints
added during the construction~of~$ \varphi_{\textit{LTL}} $.

The algorithm iterates over all possible
responses~$ v \in \comps^{m \cdot b} $ of the system up to depth
$ m \cdot b $. This is sufficient, since any deeper exploration would
result in a state repetition of the cross-product of the finite state
representation of~$ \pi $ and any system strategy bounded by~$ b
$. Hence, the same behaviour could also be generated by a smaller
sequence. At the same time, the algorithm iterates over
predicates~$ \pterm \in \pterms $ appearing in
$ \varphi_{\textit{TSL}} $ and times $ t $ and $ t' $ smaller
than~$ m \cdot b $. For each of these elements, spuriousness is
checked by comparing the output of~$ \pi $ for the evaluation of
$ \pterm $ at times~$ t $ and $ t' $, which should only
differ, if the inputs to the predicates are different as well. This
can only happen, if the passed input terms have been constructed
differently over the past. We check it by using the evaluation
function~$ \eta $ equipped with the identity assignment
$ \assign{\cdot}_{\texttt{id}} \from \fnames \to \fnames $, with
$ \assign{\name{f}}_{\texttt{id}} = \name{f} $ for all
$ \name{f} \in \fnames $, and the input sequence
$ \iota_{\texttt{id}} $, with
$ \iota_{\texttt{id}}(t)(\name{i}) = (t,\name{i}) $ for all
$ t \in \dtime$ and $ \name{i} \in \inames $, that always generates a
fresh input. Syntactic inequality of
$ \evalid(v,\iota_{\name{id}},t,\pterm) $ \linebreak and
$ \evalid(v,\iota_{\name{id}},t',\pterm) $ then is a sufficient
condition for the existence of an assignment
$ \assign{\cdot} \from \fterms \to \functions $, for which $ \pterm $
evaluates differently at times $ t $ and~$ t' $.

If spurious behaviour of~$ \pi $ could be found, then the revealing
response~$ v \in \comps^{*} $ is first simplified using
$ \texttt{reduce} $, which turns $ v $ back to a sequence of sets of
updates~$ w \in (2^{\utermsp})^{*} $ and removes updates that do not
affect the behavior of $ \pterm $ at the times $ t $ and $ t' $ to
accelerate the termination of the CEGAR loop. Afterwards, the
sequence~$ w $ is turned into a new assumption that prohibits the found
spurious behavior and, thus, further refines the LTL
underapproximation.

As an example of this process, reconsider the spurious
counter-strategy of \cref{eq:tslSimpleSoln}. Already after the first
system response~$ \upd{\name{y}}{\name{x}} $, the environment produces
an inconsistency by evaluating~$ \name{p} \ \name{x} $ and
$ \name{p} \ \name{y} $ differently. This is inconsistent, as the
cell~$ \name{y} $ holds the same value at time~$ t = 1 $ as the
input~$ \name{x} $ at time~$ t = 0 $. Using \cref{alg:spurious} we generate
the new
assumption~$ \LTLglobally (\upd{\name{y}}{\name{x}} \impl (\name{p} \
\name{x} \leftrightarrow \LTLnext \name{p} \ \name{y})) $. After adding this
strengthening the LTL synthesizer returns a realizability result.

\medskip

\goodbreak

\noindent \textit{Undecidability}
Although we can approximate the semantics of TSL with LTL, there are
TSL formulas that cannot be expressed as LTL formulas of finite
size.
\begin{theorem}\label{thm:decidability}
  The realizability problem of $ \TSL $ is undecidable.
\end{theorem}
\begin{proof}
  We reduce an instance of the Post Correspondence
  Problem~(PCP)~\cite{post1946}, consisting of an alphabet~$ \Sigma $
  and sequences
  $ w_{0}w_{1}\ldots w_{n}, v_{0}v_{1}\ldots v_{n} \in \Sigma^{*} $,
  to the realizability of a $ \TSL $ formula~$ \varphi $. To this end,
  we fix some unary predicate~$ \name{p} \in \pnames $, a~unary
  function $ \name{f} \in \fnames $ for every alphabet symbol
  $ f \in \Sigma $, and some $ 0 $-nary
  function~$ \name{X} \in \fnames $. The system has no
  inputs~$ \inames $, but two outputs $ \name{A} \in \onames $ and
  $ \name{B} \in \onames $.

  Initially, we assign the signals $ \name{A} $ and $ \name{B} $ the
  constant value~$ \name{X} $. From then on, we non-deterministically
  pick pairs $ (w_{j},v_{j}) $ in every time step, as provided by the
  PCP instance, where every $ w_{j} $ and $ v_{j} $ is represented as
  a stacked composition of the corresponding alphabet functions. Our
  choice is stored in the signals~$ \name{A} $ and $ \name{B} $ for
  $ w_{j} $ and $ v_{j} $, respectively. Finally, we check that the
  sequences of function applications, constructed over time, are equal
  at some point, using the eventually operator~$ \LTLfinally $ and the
  universally quantified predicate~$ \name{p} $ to check for equality.
\end{proof}
\noindent A more detailed version of the proof can be found in
\cref{proof:decidability}. Also note that no inputs are used by
the proof, which additionally shows that the \mbox{``satisfiability''} problem of
\TSL is undecidable as well.

\section{TSL Synthesis}
\label{sec:synth}
Our synthesis framework provides a modular refinement process
to synthesize executables from $ \TSL $ specifications, as depicted
in \cref{fig:system}. The user initially provides a
$ \TSL $ specification over predicate and function terms.  At the end
of the procedure, the user receives an executable to control a
reactive system.

The first step of our method answers the synthesis question of TSL: if
the specification is realizable, then a control flow model is
returned.  To this end, an intermediate translation to LTL is used,
utilizing an LTL synthesis solver that produces circuits in the AIGER
format. If the specification is realizable, the resulting control flow
model is turned into Haskell code, which is implemented as an
independent Haskell module. The user has the choice between two
different targets: a module built on Arrows, which is compatible with
any Arrowized FRP library, or a module built on Applicative, which
supports Applicative FRP \mbox{libraries}. Our procedure generates a single
Haskell module per TSL specification. This makes naturally decomposing
a project according to individual tasks possible. Each module provides
a single component, which is parameterized by their initial state and
the pure function and predicate transformations. As soon as these are
provided as part of the surrounding project context, a final
executable can be generated by compiling the Haskell code.

An important feature of our synthesis approach is that implementations
for the terms used in the specification are only required after
synthesis.  This allows the user to explore several possible
specifications before deciding on any term implementations.

\paragraph{Control Flow Model} The first step of our approach is the
synthesis of a \textit{Control Flow Model}~$ \cfm $ (CFM) from the
given $ \TSL $ specification~$ \varphi $, which provides us with a
uniform representation of the control flow structure of our final
program.

\noindent Formally, a CFM~$ \cfm $ is a tuple
$ \cfm = (\inames, \onames, \cells, \vertices, \labeling,
\dependencies), $ where $ \inames $ is a finite set of inputs,
$ \onames $ is a finite set of outputs, $ \cells $ is a finite set of
cells, $ \vertices $ is a finite set of vertices,
$ \labeling \from \vertices \to \fnames $ assigns a
vertex a function~$ \name{f} \in \fnames $ or a
predicate~$ \name{p} \in \pnames $, and
\begin{equation*}
  \dependencies \from (\onames \cup \cells \cup \vertices) \times
  \nats \to (\inames \cup \cells \cup \vertices \cup \set{ \bot })
\end{equation*}
is a dependency relation that relates every output, cell, and
vertex of the CFM with $ n \in \nats $ arguments, which are either
inputs, cells, or vertices. Outputs and
cells~$ \name{s} \in \onames \cup \cells $ always have only a single
argument, i.e., $ \delta(s, 0) \not\equiv \bot $ and
\mbox{$ \forall m > 0 .\ \delta(\name{s}, m) \equiv \bot $}, while for
vertices~$ x \in \vertices $ the number of arguments $ n \in \nats $
align with the arity of the assigned function or predicate
$ \labeling(x) $, i.e.,
$ \forall m \in \nats .\ \delta(s, m) \equiv \bot \leftrightarrow m >
n $. A CFM is valid if it does not contain circular dependencies,
i.e., on every cycle induced by $ \delta $ there must lie at least a
single cell. We only consider valid CFMs.

\newcommand{\arrow}[3]{
  \node[anchor=west,inner sep=0pt] at #1 (#2) {
    \begin{tikzpicture}[scale=0.8,inner sep=2pt]
      \node[arrow] (A) at (-3,0) {#3};      
      \fill[green!10] (A.south east) -- (A.south west) -- 
      (A.north west) -- (A.north east) 
      -- ($ (A.east) + (0.2,0) $) -- cycle;
      \node[arrow] (A) at (-3,0) {\phantom{#3}};      
      \draw (A.south east) -- (A.south west) -- 
      (A.north west) -- (A.north east) -- 
      ($ (A.east) + (0.2,0) $) -- cycle;
      \node[text depth=0pt] at (A) at (-3,0) {#3};            
    \end{tikzpicture}
  };
}
\begin{figure}[t]
  \centering

\begin{tikzpicture}[circuit logic US]
  \tikzset{ 
    arrow/.style={fill=green!10},      
    and/.style={and gate, inputs={nn}, point right,blue!80!black!50!white,fill},
    or/.style={or gate, inputs={nn}, point right,blue!70,fill},
    not/.style={not gate, point right, scale=0.5,black!50!blue, fill},
    sec/.style={fill, circle,inner sep=0.7pt}, 
  }
  
  \draw[drop shadow,fill=gray!10] (-5.5,-4.6) rectangle (4.25,2);

  \begin{scope}[circuit logic US,line width=0.4,scale=0.45,xshift=-81,yshift=55]
    \tikzset{
      and/.style={and gate, inputs={nn}, point right,blue!80!black!50!white,fill},
      or/.style={or gate, inputs={nn}, point right,blue!70,fill},
      not/.style={not gate, point right, scale=0.5,black!50!blue, fill},
      sec/.style={fill, circle,inner sep=0.7pt},
    }
    \draw[fill=yellow!10,rounded corners=3,fill,drop shadow,thin] (-2,-5.9) rectangle (7.8,2);

    \node at (-2,-2.6) (pause) {}; 
    \node at (-2,1.4) (cfg) {}; 
    \node at (-2,0.7) (play) {}; 
    \node at (-2,-1.3) (resume) {};
    \node at (-2,-5.1) (leave) {};  
    \node at (-2,-5.3) (music) {}; 

    \node at (5.7,-5.9) (out1) {};
    \node at (6.5,-5.9) (out2) {};
    \node at (7.3,-5.9) (out3) {};

    \node at (7.7,0.9) (out4) {};
    \node at (7.7,1.4) (out5) {};

    \node[and] at (0.05,0.6) (a1) {};
    \node[or] at (0,-1.4) (o1) {};
    \node[or] at (0,-2.2) (o2) {};
    \node[or] at (0,-3) (o3) {};
    \node[not] at (1.1,-3) (n1) {};
    \node[not] at (1.1,-2.2) (n2) {};
    \node[not] at (1.1,0.2) (n3) {};
    \node[and] at (2.3,-1) (a2) {};
    \node[or] at (3.6,-0.5) (o4) {};
    \node[and] at (2.3,-3.7) (a3) {};
    \node[and] at (0.05,-5.2) (a4) {};
    \node[and] at (2.3,-4.5) (a5) {};
    \node[or] at (3.6,-4.1) (o5) {};
    \node[not] at (1.1,-5.2) (n4) {};
    \node[and] at (3.7,-3.1) (a6) {};
    \node[and] at (5.3,-1.8) (a7) {};
    \node[or] at (6.6,0.1) (o6) {};
    \node[or] at (5.2,0.9) (o7) {};
    \node[not] at (6.6,0.9) (n5) {};

    \draw (o1.output) ++ (right:0.2) node[sec] {} |- (n3.input);
    \draw (o1.output) -- ++ (right:0.2) |- (a3.input 1);
    \draw (o2.output) -- (n2.input);
    \draw (o3.output) -- (n1.input);
    \draw (n2.output) -- ++ (right:0.2) |- (a2.input 2);
    \draw (a2.output) -- ++ (right:0.2) |- (o4.input 2);
    \draw (a1.output) -| ($ (a2.output) + (0.2,1) $) |- (o4.input 1);
    \draw (a4.output)  ++ (right:0.25) node[sec] {} |- (a5.input 2);
    \draw (a4.output) -- (n4.input);
    \draw (a3.output) -- ++ (right:0.2) |- (o5.input 1);
    \draw (a5.output) ++ (right:0.2) |- (o5.input 2);
    \draw (a5.output) -- ++ (right:0.2) node[sec] {} |- ($ (a6.output) + (0.2,-1.7) $) |- (o7.input 2);
    \draw (n1.output) -- (a6.input 1);
    \draw (n3.output) ++ (right:0.8) node[sec] {} |- (o7.input 1);
    \draw (n3.output) -- (o6.input 1);
    \draw (n4.output) -- ++ (right:0.2) |- (a6.input 2);
    \draw (a6.output) -- ++ (right:0.5) node (A) {} |- (a7.input 2);
    \draw (a7.output) -- ++ (right:0.2) |- (o6.input 2);
    \draw (o7.output) -- (n5.input);

    \draw (o5.output) -| (out1.center);
    \draw (o4.output) -- ++ (right:2.4) |- (out2.center);
    \draw (o6.output) -- ++ (right:0.2) |- (out3.center);
    \draw (n5.output) -- (out4.center);
    \draw (o7.output) ++ (right:0.2) node[sec] {} |- (out5.center);

    \draw (pause.center) -- ++ (right:0.5) |- (o2.input 1);
    \draw (pause.center) ++ (right:0.5) node[sec] (T) {} |- node[sec] {} (o3.input 2);
    \draw (T.center |- o3.input 2) |- (a3.input 2);

    \draw (cfg.center) -- ++ (right:1.3) node[sec] (U) {} |- node[sec] {} (a1.input 2);
    \draw (U.center |- a1.input 2) |- node[sec] {} (o1.input 2);
    \draw (U.center |- o1.input 2) |- node[sec] {} (o2.input 2);
    \draw (U.center |- o2.input 2) |- (a5.input 1);
    \draw (cfg.center) -| ($ (A.center) + (0,4) $)  |- (a7.input 1);

    \draw (play.center) -- (a1.input 1);
    \draw (play.center) ++ (right:0.9) node[sec] {} |- (o3.input 1);

    \draw (resume.center) ++ (right:0.5) node[sec] {} |- (a2.input 1);
    \draw (resume.center) -- (o1.input 1);

    \draw (leave.center) -- ++ (right:0.5) |- (a4.input 1);

    \draw (music.center) -- (a4.input 2);
  \end{scope}
  
  \node at (-5,0) (Y) {};
  \node at (1,-2.1) (Z) {};

  \node at (-5.5,-2.6) (mpin) {};
  \node at (-5.5,0.05) (sys) {};
  \node at (-5.5,-4.2) (tr) {};
  \node at (tr.center |- cfg.center) (cin) {};
  \node at (4.25,0.45) (cout) {};
  \node at (4.25,-3.3) (mpout) {};

  \arrow{(Y.center |- cfg.center)}{nA}{$ \name{(== m}_{\name{0}} \!\name{)} $};
  \arrow{(-5,0.9)}{nB}{$ \name{playButton} $};
  \arrow{(Y.center |- resume.center)}{nC}{$ \name{resumeApp} $};
  \arrow{(Y.center |- pause.center)}{nD}{$ \name{pauseButton} $};
  \arrow{(-5,-0.8)}{nE}{$ \name{leaveApp} $};
  \arrow{(Y.center |- music.center)}{nF}{$ \name{musicPlaying} $};

  \arrow{(2.5,0.7)}{nG}{$ \name{m}_{\name{0}} $};
  \arrow{(2.5,0.2)}{nH}{$ \name{m}_{\name{1}} $};
  \arrow{(-2.75,-2.6)}{nI}{\tikz{\node[inner sep=3pt,minimum height=1em]{$ \name{pause} $};}};
  \arrow{(-4.5,-3.8)}{nJ}{$ \name{trackPos} $};

  \arrow{(-1.3,-4)}{nK}{\tikz{\node[inner sep=5pt,minimum height=2em]{$ \name{play} $};}};

  \draw[fill=blue!10] (3.3,0) rectangle (4,1.1);
  \node[circle,draw,fill=white] at (3.75,0.45) (S1) {};

  \draw[fill=blue!10] (1,-4.3) rectangle (2.3,-2.1);
  \node[circle,draw,fill=white] at (2,-3.3) (S2) {};

  \draw (nA.east) ++(left:0.01) -- (cfg.center);
  \draw (nB.east) ++(left:0.01) -- ++(right:0.4) |- (play.center);
  \draw (nC.east) ++(left:0.01) -- (resume.center);
  \draw (nD.east) ++(left:0.01) -- (pause.center);
  \draw (nE.east) ++(left:0.01) -- ++(right:0.9) |- (leave.center);
  \draw (nF.east) ++(left:0.01) -- (music.center);

  \draw (nG.east) ++(left:0.01) -- ++(right:0.4) node[xshift=-4,yshift=3] {\tiny 1} -- (S1);
  \draw (nH.east) ++(left:0.01) -- ++(right:0.4)  node[xshift=-4,yshift=3] {\tiny 2} -- (S1);

  \draw (out4.center) -| (3.5,1.1) node[below,yshift=2] {\tiny 1};
  \draw (out5.center) -| (3.8,1.1) node[below,yshift=2] {\tiny 2};

  \draw (nJ.east) ++(left:0.01) -- (nJ.east -| nK.west) -- ++(right:0.01);

  \draw (out1.center) -- (Z.center -| out1.center) node[yshift=-3] {\tiny 1};
  \draw (out2.center) -- (Z.center -| out2.center) node[yshift=-3] {\tiny 2};
  \draw (out3.center) -- (Z.center -| out3.center) node[yshift=-3] {\tiny 3};

  \draw (nI.east) ++(left:0.01) -- ++ (right:2.7) node (W) {} -- (S2);
  \draw (nK.east) ++(left:0.01) -- (nK.east -| W.center) -- (S2);

  \draw (mpin.center) -- ++ (right:0.25) node[sec] (Q) {} |- (S2);
  \draw (Q |- S2) node[sec] {} |- (nJ.west) -- ++(right:0.01);
  \draw (Q) -- (nI.west) -- ++(right:0.01);
  \draw (Q |- nI.west) |- (nF.west) -- ++(right:0.01);

  \draw (tr.center) -- (tr -| nK.west) -- ++(right:0.01);
  \draw (Q -| Z) node[xshift=3,yshift=3] {\tiny 1};
  \draw (S2 -| Z) node[xshift=3,yshift=3] {\tiny 3};
  \draw (nK.east  -| Z) node[xshift=3,yshift=3] {\tiny 2};

  \draw (sys.center) -- ++ (right:0.25) node[sec] (L) {} |- (nD.west) -- ++(right:0.01);
  \draw (L.center) |- (nC.west) -- ++(right:0.01);
  \draw (L.center |- nD) node[sec] {} |- (nE.west) -- ++(right:0.01);
  \draw (L.center |- nC) node[sec] {} |- (nB.west) -- ++(right:0.01);

  \draw (cin.center) -- (nA.west) -- ++(right:0.01);

  \draw (S1) -- (cout.center);
  \draw (S2) -- (mpout.center);

  \draw[->,>=stealth,thick] (cout.center) -- ++(right:0.9) -- ++(up:2) -- ++(left:11.55) |- (cin.center);
  \draw[->,>=stealth,thick] (mpout.center) -- ++(right:0.9);

  \draw[<-,>=stealth,thick] (mpin.center) -- ++(left:0.8);
  \draw[<-,>=stealth,thick] (tr.center) -- ++(left:0.8);
  \draw[<-,>=stealth,thick] (sys.center) -- ++(left:0.8);
  
  \draw (cout.north east) node[xshift=-2,yshift=1,anchor=west] {\scalebox{0.8}{$ \name{Cell} $}};
  \draw (mpout.north east) node[xshift=-3,yshift=1,anchor=west] {\scalebox{0.8}{$ \name{Ctrl} $}};

  \draw (mpin.north west) node[xshift=2,yshift=1,anchor=east] {\scalebox{0.8}{$ \name{MP} $}};
  \draw (cin.north west) node[xshift=2,yshift=1,anchor=east] {\scalebox{0.8}{$ \name{Cell} $}};
  \draw (tr.north west) node[xshift=2,yshift=1,anchor=east] {\scalebox{0.8}{$ \name{Tr} $}};
  \draw (sys.north west) node[xshift=2,yshift=1,anchor=east] {\scalebox{0.8}{$ \name{Sys} $}};

  \draw[gray!80,fill=white,rounded corners=3] (2.5,-2.55) rectangle (4.1,-0.55);

  \node[and,scale=0.45,anchor=west] at (2.7,-0.8) {};
  \node at (3.25,-0.8) {$ \equiv $};
  \node[anchor=west] at (3.35,-0.8) {\scalebox{0.8}{\name{and}}};

  \node[or,scale=0.45,anchor=west] at (2.7,-1.3) {};
  \node at (3.25,-1.3) {$ \equiv $};
  \node[anchor=west] at (3.35,-1.3) {\scalebox{0.8}{\name{or}}};

  \node[not,scale=0.45,anchor=west] at (2.7,-1.8) {};
  \node at (3.1,-1.8) {$ \equiv $};
  \node[anchor=west] at (3.2,-1.8) {\scalebox{0.8}{\name{not}}};

  \node[circle,draw,fill=white,anchor=west,scale=0.9] at (2.6,-2.3) {};
  \node at (3.1,-2.3) {$ \equiv $};
  \node[anchor=west] at (3.2,-2.3) {\scalebox{0.8}{\name{mutex}}};
\end{tikzpicture}

\caption{Example CFM of the music player generated from a TSL
  specification.}
\label{fig:cfmexample}
\vspace{-1em}
\end{figure}
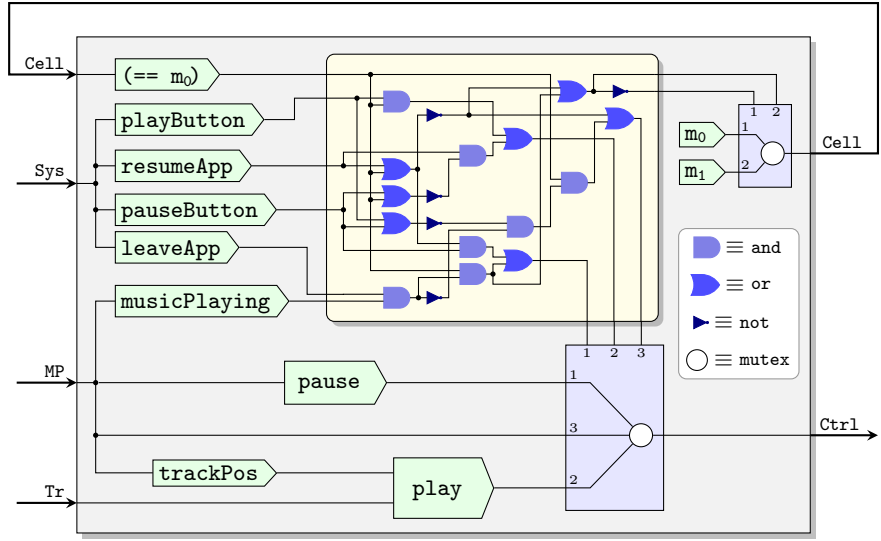

An example CFM for our music player of \cref{sec:motiv} is depicted in
\cref{fig:cfmexample}. Inputs~$ \inames $ come from the left
and outputs~$ \onames $ leave on the right. The example
contains a single cell~$ \name{c} \in \cells $, which holds the
stateful memory~\name{Cell}, introduced during synthesis for the
module. The green, arrow shaped boxes depict vertices~$ \vertices $,
which are labeled with functions and predicates names, according
to~$ \labeling $. For the Boolean decisions that define $\delta$, we use circuit symbols for
conjunction, disjunction, and negation. Boolean decisions are piped to
a multiplexer gate that selects the respective
update streams. This allows each update stream to be passed to an output stream if and only if the
respective Boolean trigger evaluates positively, while
our construction ensures mutual exclusion on the Boolean triggers. For
code generation, the logic gates are implemented using the corresponding dedicated Boolean functions.
After building a control structure, we assign semantics to functions
and predicates by providing implementations.  To this end, we use
Functional Reactive Programming (FRP).  Prior work has established
Causal Commutative Arrows (CCA) as an FRP language pattern equivalent
of a CFM~\cite{jfp/LiuCH11,liu2007plugging,yallop2016causal}.  CCAs
are an abstraction subsumed by other functional reactive programming
abstractions, such as Monads, Applicative and
Arrows~\cite{jfp/LiuCH11,lindley2011idioms}.
There are many FRP libraries using
Monads~\cite{elm,hudakFRAN,ploeg2015frpnow},
Applicative~\cite{reactivebanana,clash2015,helbling2016juniper,Reflex},
or Arrows~\cite{courtney2003yampa,murphy2016livefrp,perez2016yampa,UISF},
and since every Monad is also an Applicative and Applicative/Arrows both are universal design patterns, we can give uniform
translations to all of these libraries using translations to just Applicative
and Arrows. Both translations are possible due to the flexible notion of a CFM.

In the last step, the synthesized FRP program is compiled into an
executable, using the provided function and predicate
implementations. This step is not fixed to a single
compiler implementation, but in fact can use any FRP compiler (or
library) that supports a language abstraction at least as expressive as CCA.
For example, instead of creating an Android music player app, we could
target an FRP web interface~\cite{Reflex} to create an online music
player, or an embedded FRP library~\cite{helbling2016juniper} that
allows us to directly instantiate the player on a computationally more
restricted device. By using the strong core of CCA, we even can go
down the whole chain and directly implement the player in hardware,
which is for example possible with the C$ \lambda $aSH  compiler~\cite{clash2015}.
Note that we still need to give separate implementations for the functions and
predicates for each target. However, our specification and the
synthesized CFM always stay the same.

\section{Experimental Results}
\label{sec:eval}
To evaluate our synthesis procedure we implemented a tool that follows
the structure of \cref{fig:system}.  Our tool first encodes the given
\TSL specification in LTL and then refines it until an LTL solver
either produces a realizability result or returns a counter-strategy
that is non-spurious. For LTL synthesis we use the bounded synthesis
tool BoSy~\cite{bosy}. As soon as we get a realizing strategy, given
as a circuit, it is translated to a corresponding CFM. Then, we
generate the FRP program structure. Finally, after providing function
and predicate implementations the resulting program is compiled into
an executable.

\begin{table}[htbp]
\centering
\caption{Number of cells~$ |\cells_{\cfm}| $ and
  vertices~$ |V_{\cfm}|$ of the resulting CFM~$ \cfm $ and
  synthesis times for a collection of \TSL specifications~$ \varphi
  $. A * indicates that the benchmark additionally has an initial
  condition as part of the specification.}
\label{table:results}
\begin{tabular}{|l||c|c|c|c|c||c|c|c|c|}
\hline
\multicolumn{1}{|c||}{\multirow{2}{*}{\textsc{Benchmark} $ (\varphi) $}}
  & \multicolumn{1}{c|}{\multirow{2}{*}{\ \,$| \varphi |$\ \,}}
  & \multicolumn{1}{c|}{\multirow{2}{*}{\ \,$| \inames |$\ \,}}
  & \multicolumn{1}{c|}{\multirow{2}{*}{\ \,$| \onames |$\ \,}}
  & \multicolumn{1}{c|}{\multirow{2}{*}{\ \,$| \pnames |$\ \,}}
  & \multicolumn{1}{c||}{\multirow{2}{*}{\ \,$| \fnames |$\ \,}}
  & \multicolumn{1}{c|}{\multirow{2}{*}{\ $| \cells_{\cfm} |$\ }}
  & \multicolumn{1}{c|}{\multirow{2}{*}{\ $| V_{\cfm} |$\ }}
  & \multicolumn{1}{c|}{\raisebox{-2pt}{\ \textsc{Synthesis}\ }}
  \\
&&&&&&&&\raisebox{0pt}{\textsc{Time (s)}}\\
\hline \hline
\textbf{Button} &&&&&&&& \\
\ \ \ default & 7 & 1 & 2 & 1 &  3 & 3 & 8 & 0.364 \\
\hline
\textbf{Music App} &&&&&&&& \\
\ \ \ simple & 91 & 3 & 1 & 4 & 7 & 2 & 25 & 0.77 \\ 
\ \ \ system feedback & 103 & 3 & 1 & 5 & 8 & 2 & 31 & 0.572 \\ 
\ \ \ motivating example \mbox{\quad\ \,} & 87 & 3 & 1 & 5 & 8 & 2 & 70 & 1.783 \\ 
\hline
\textbf{FRPZoo} &&&&&&&& \\
\ \ \ scenario\hspace{0.5pt}$_{0} $ & 54 & 1 & 3 & 2 & 8 & 4 & 36 & 1.876 \\ 
\ \ \ scenario\hspace{0.5pt}$_{5} $ & 50 & 1 & 3 & 2 & 7 & 4 & 32 & 1.196 \\ 
\ \ \ scenario\hspace{0.5pt}$_{10} $ & 48 & 1 & 3 & 2 & 7 & 4 & 32 & 1.161 \\ 
\hline
\textbf{Escalator} &&&&&&&& \\
\ \ \ non-reactive & 8 & 0 & 1 & 0 & 1 & 2 & 4 & 0.370 \\ 
\ \ \ non-counting & 15 & 2 & 1 & 2 & 4 & 2 & 19 & 0.304 \\ 
\ \ \ counting & 34 & 2 & 2 & 3 & 7 & 3 & 23 & 0.527 \\ 
\ \ \ counting* & 43 & 2 & 2 & 3 & 8 & 4 & 43 & 0.621 \\ 
\ \ \ bidirectional & 111 & 2 & 2 & 5 & 10 & 3 & 214 & 4.555 \\ 
\ \ \ bidirectional* &\,124\,& 2 & 2 & 5 & 11 & 4 & 287 & 16.213 \\ 
\ \ \ smart & 45 & 2 & 1 & 2 & 4 & 4 & 159 & 24.016 \\ 
\hline
\textbf{Slider} &&&&&&&& \\
\ \ \ default & 50 & 1 &  1 & 2 & 4 & 2 & 15 & 0.664 \\ 
\ \ \ scored & 67 & 1 & 3 & 4 & 8 & 4 & 62 & 3.965 \\ 
\ \ \ delayed & 71 & 1 & 3 & 4 & 8 & 5 & 159 & 7.194 \\ 
\hline
  \textbf{Haskell-TORCS}  &&&&&&&& \\
\ \ \ simple & 40 & 5 & 3 & 2 & 16 & 4 & 37 & 0.680 \\
\ \ \ \textbf{advanced} &&&&&&&& \\[-0.2em]
\ \ \ \ \ \ gearing & 23 & 4 & 1 & 1 & 3 & 2 & 7 & 0.403 \\
\ \ \ \ \ \ accelerating & 15 & 2 & 2 & 2 & 6 & 3 & 11 & 0.391 \\
\ \ \ \ \ \ \textbf{steering} &&&&&&&& \\[-0.2em]
\ \ \ \ \ \ \ \ \ simple & 45 & 2 & 1 & 4 & 6 & 2 & 31 & 0.459 \\
\ \ \ \ \ \ \ \ \ improved & 100 & 2 & 2 & 4 & 10 & 3 & 26 & 1.347 \\
\ \ \ \ \ \ \ \ \ smart & 76 & 3 & 2 & 4 & 8 & 5 & 227 & 3.375 \\
\hline
\end{tabular}
\end{table}

\begin{table}[htbp]
\centering
\caption{Set of programs that use purity to keep one or two counters
  in range. Synthesis needs multiple refinements of the specification
  to proof realizability.}
\label{table:results2}
\begin{tabular}{|l||c|c|c|c|c||c|c|c|c|}
\hline
  \multicolumn{1}{|c||}{\multirow{2}{*}{\ \textsc{Benchmark} $ (\varphi) $\ }}
  & \multicolumn{1}{c|}{\multirow{2}{*}{\;$| \varphi |$\;}}
  & \multicolumn{1}{c|}{\multirow{2}{*}{\;$| \inames |$\;}}
  & \multicolumn{1}{c|}{\multirow{2}{*}{\;$| \onames |$\;}}
  & \multicolumn{1}{c|}{\multirow{2}{*}{\;$| \pnames |$\;}}
  & \multicolumn{1}{c||}{\multirow{2}{*}{\;$| \fnames |$\;}}

  & \multicolumn{1}{c|}{\multirow{2}{*}{\,$| \cells_{\cfm} |$\,}}
  & \multicolumn{1}{c|}{\multirow{2}{*}{\,$| V_{\cfm} |$\,}}
  & \multicolumn{1}{c|}{\multirow{2}{*}{\textsc{Refinements}}}
  & \multicolumn{1}{c|}{\raisebox{-2pt}{\;\textsc{Synthesis}\;}}
  \\
&&&&&&&&&\raisebox{0pt}{\textsc{Time (s)}}\\
\hline \hline
\ inrange-single  & 23 & 2  & 1 & 2 & 4 & 2 & 21 & 3 & 0.690 \\ %
\ inrange-two  & 51 & 3 &  3 & 4 & 7 & 4 & 440 & 6 & 173.132 \\ %
\ graphical-single  & 55 & 2 & 3 & 2 & 6 & 4 & 343 & 9 & 1767.948 \\ %
\ graphical-two  & 113 & 3 &  5 & 4 & 9 & -  & - & - & >\,10000 \\ %
\hline
\end{tabular}
\end{table}

To demonstrate the effectiveness of synthesizing \TSL, we applied our
tool to a collection of benchmarks from different application domains,
listed in \cref{table:results}.  Every benchmark class consists of
multiple specifications, addressing different features of \TSL.  We
created all specifications from scratch, where we took care that they
either relate to existing textual specifications, or real world
scenarios. A short description of each benchmark class is given in
\cref{apx:benchmarks}.

For every benchmark, we report the synthesis time and the size of the
synthesized CFM, split into the number of cells
($ | \cells_{\cfm} | $) and vertices ($ | V_{\cfm} | $) used.
The synthesized CFM may use more cells than the original TSL
specification if synthesis requires more memory in order to realize a
correct control flow.
The synthesis was executed on a quad-core Intel Xeon processor
(E3-1271 v3, 3.6GHz, 32 GB RAM, PC1600, ECC), running
Ubuntu 64bit LTS 16.04.

The experiments of \cref{table:results} show that TSL successfully lifts the applicability of
synthesis from the Boolean data domain to arbitrary data domains,
which allows for new applications that can utilize every level of
required abstraction. For the benchmarks, we could always find a
realizable system within a reasonable amount of time, where the
solutions found often required synthesized cells to realize
the underlying control flow behavior.

We also considered a preliminary set of benchmarks that require
multiple refinement steps to be synthesizable. An overview of the
results is given in \mbox{\cref{table:results2}}. The benchmarks are inspired
by examples of the Reactive Banana FRP library~\cite{reactivebanana}.
Here, purity of function and predicate applications must be utilized
by the system to ensure that the value of one or two counters never
goes out of range. Thereby, the system not only needs purity to verify
this condition, but also to take the correct decisions in the
resulting implementation to be synthesized.

\section{Related Work}
\label{sec:related}
Synthesizing reactive programs has been explored in domains such as imperative programs over finite domains~\cite{madhusudan2011synthesizing} and parallel execution strategies for sequential programs~\cite{conf/fmcad/BloemHKKAS14}.
There have also been many alternatives to LTL for specifying properties of reactive systems, such as $ \mu$-calculus~\cite{Kupferman:2000}, Signal Temporal Logic~\cite{hscc/sanjit15}, Ground Temporal Logic~\cite{cyrluk1994ground}, and coalgebraic logics~\cite{bonsangue2008coalgebraic}.
An example of a synthesis approach that integrates control and data is recent work on strategy synthesis for linear arithmetic games~\cite{Farzan:2017}.
While reactive synthesis has focused on the complex control aspects of reactive systems,
  deductive and inductive synthesis has been concerned with the data transformation aspects in non-reactive and sequential programs~\cite{vechevYY13,kuncak2010complete,osera2015type,solarLezama13,Feser:2015:SDS:2737924.2737977,Isil17}.
Abstraction-based approaches to reactive synthesis~\cite{Beyene:2014:CAS:2535838.2535860,Dimitrova2012,hsu2018multi,mallik2016compositional} can be seen as a link between deductive and reactive synthesis.
In terms of FRP, a Curry-Howard correspondence between LTL and FRP in a dependently typed language was discovered~\cite{plpv/Jeffrey12,jeltsch2012towards} and used to prove properties of FRP programs~\cite{Cave2014Fair,krishnaswami2013higher}.

\section{Conclusions}
\label{sec:conclusions}
We have introduced Temporal Stream Logic, which allows the user to
specify the control flow of a reactive program.  The logic cleanly
separates control from complex data, forming the foundation for our
procedure to synthesize FRP programs. By utilizing the purity of
function transformations our logic scales independently of the
complexity of the data to be handled. While we have shown that
scalablility comes at the cost of undecidability, we addressed this
issue by using a CEGAR loop, which lazily refines the
underapproximation until either a realizing system implementation or
an unrealizability proof is found.

TSL also provides the foundations for further extensions. For example,
a user may want to fix the semantics for a partial set of functions and
predicates to be utilized by the synthesis tool. Such additional
refinements can be implemented as part of a much richer \textit{TSL
  Modulo Theory} framework.

Our experiments indicate that TSL synthesis works well in practice and on
a wide range of programming applications.  In general, we hope the
applications of this new logic and approach to reactive synthesis will
stimulate further research into the scalable, real-world use of
temporal logics for synthesis.

\bibliographystyle{splncs}
\bibliography{biblio}

\newpage

\appendix
\section{Appendix}
\subsection{Specifying a Music Player in TSL}
\label[appendix]{apx:musicspec}

To demonstrate the simplicity of TSL, we illustrate the creation of a TSL specification for the music player Android app of \cref{sec:motiv}.
For the concrete implementation, we use the \texttt{MediaPlayer} class (\name{MP}) from the Android API, which provides functions to pause and play, as well as a predicate to check if music is currently playing or not.
Specific to the Android OS, we receive signals when a user leaves or resumes the app.
Specific to our particular app, we also use two buttons in the UI that deliver signals when a user presses play or pause.

In contrast to the FRP model, the Android system uses callback structures and functions have side effects, such as playing the music.
Although the Android code is not using an FRP model, the theoretical foundations provided by FRP make embedding the synthesized control code a straightforward task.

Nevertheless, we need to consider that libraries used by the app and the surrounding android system carry their own state, which needs to be reflected within the model. To do so, we introduce a separate stream for each interface, which we assume to carry all the necessary state. We use the input stream~\name{Sys} to receive system events and button presses, while the input stream~$ \name{MP} $ provides us the interface to the \texttt{MediaPlayer} class. Updates to the \texttt{MediaPlayer} class are provided via its control interface~$ \name{Ctrl} $. This allows us to embed the synthesized program into any larger context to manipulate the music player. We cannot use the same name for input and output here, as we would miss state changes that could be made by a component plugged after this, which is a specific result of the clear modularization utilized within CCA. This is also why we do not need a separate system output here, as we only receive signals from the system. Finally, we also utilize the input stream~\name{Tr}, which provides us with the currently selected music track.

We partition the individual properties of the specification into two
categories: assumptions~$ A $ and guarantees~$ G $. The final
specification then results from their implication:
$ \bigwedge_{\vartheta \in A} \vartheta \rightarrow \bigwedge_{\psi
  \in B} \psi $. Our specification uses the following signals,
function terms and predicates:
\begin{align*}
  \inames &= \{\, \name{MP}, \name{Sys},  \name{Tr}\,\} \\
  \onames &= \{\, \name{Ctrl} \,\} \\
  \cells  &= \{\, \name{Ctrl} \,\} \\
  \pterms &= \{\, \name{musicPlaying},\, \name{pauseButton},\name{playButton}, \name{leaveApp}, \name{resumeApp} \,\} \\
  \fterms &= \pterms \cup \{\, \name{pause},\, \name{play},\name{trackPos} \,\} \\
\end{align*}
The function~$ \name{pause}~\name{m} $ pauses a played track on the passed music player stream~$ \name{m} $. The function~$\name{play}~\name{t}~\name{p} $ plays the selected track~$ \name{t} $ at the given track position~$ \name{p} $. The track position is carried by the music player stream, and can be extracted using the function~\name{trackPos}. In our model, \name{play} resets any state that is passed by $ \name{MP} $, while \name{pause} does not. The predicate \name{musicPlaying} checks whether music is playing on a music player stream, while the remaining predicates check the respective conditions from the surrounding android system.

\para{Assumptions}
We start with straightforward assumptions about the user interface.
In our model, we assume that the pause and play buttons cannot be pressed at the same time.
Also, from the Android OS behavior, we know the app cannot leave and resume at the same time.
\setcounter{equation}{0}
\renewcommand{\theequation}{A\arabic{equation}}
\begin{align}
  &\LTLglobally \neg \big( \name{playButton}~\name{Sys} \ \wedge \ \name{pauseButton}~\name{Sys}\big) \\
  &\LTLglobally \neg \big( \name{leaveApp}~\name{Sys} \phantom{\name{o}ri} \wedge \ \name{resumeApp}~\name{Sys}\big)
\end{align}
Again from the Android OS, once a user has left the app it is not possible to press the play or pause buttons until the user resumes the app.
\begin{equation}
  \begin{array}{l}\LTLglobally \Big( \name{leaveApp}~\name{Sys} \impl \\
    \quad \hspace{2em} \big(\big(\neg \name{playButton}~\name{Sys} \wedge \neg \name{pauseButton}~\name{Sys} \big) \LTLweakuntil \name{resumeApp}~\name{Sys}\big) \Big) \end{array}
\end{equation}
We use the $ \name{musicPlaying} $ predicate to monitor changes according to the play and pause actions.
Technically, $ \name{musicPlaying} $ is not necessary to specify correct behavior - we could remember the music playing state on a separate looping stream.
However, the method is provided by the \texttt{MediaPlayer} interface and it helps to improve readability, which is why we also use it for our the specification. To obtain a correct behavior, we do not need to mimic the full behavior of the method's implementation. It suffices to specify the behavior with respect to the pause and play actions.
\begin{align}
   \LTLglobally \Big(\hspace{3.4pt}U_{\name{play}} & \to \LTLnext \big(\phantom{\neg\,} \name{musicPlaying}~\name{MP}\, \LTLweakuntil \ U_{\name{pause}} \big) \Big) \\
  \LTLglobally \Big(U_{\name{pause}}                   & \to \LTLnext \big(\neg \, \name{musicPlaying}~\name{MP}\, \LTLweakuntil \ U_{\name{play}} \hspace{3.4pt}\big) \Big)
\end{align}
We use $ U_{\name{play}} $ as a shortcut for $ \upd{\name{Ctrl}}{\name{play}~\name{Tr}~(\name{trackPos}~\name{MP})} $ and $ U_{\name{pause}} $ as a shortcut for $ \upd{\name{Ctrl}}{\name{pause}~\name{MP}} $.

\para{Guarantees}
To specify the desired behavior of our app under the given assumptions, we define the following \TSL guarantees.
First, whenever the user presses one of the buttons in the app, the output signal has to take the corresponding action.
This desired behavior is what necessitates the assumption that both buttons cannot be pressed at the same time.
Removing the assumption that both buttons cannot be pressed at the same time would require to relax these guarantees.
\setcounter{equation}{0}
\renewcommand{\theequation}{G\arabic{equation}}
\begin{align}
  \LTLglobally \big( \phantom{a} \name{playButton}~\name{Sys} & \to U_{\name{play}} \hspace{3.4pt} \big)  \\
  \LTLglobally \big( \name{pauseButton}~\name{Sys} & \to U_{\name{pause}} \big)
\end{align}
The only way that the music can be paused is either by the user leaving the app or pressing pause.
In the latter case, the music should not start playing again until either the user resumes the app or presses play.
\begin{align}
  \LTLglobally \big(  \hspace{47.02pt}  U_{\name{pause}} &\to \big(\, \name{leaveApp}~\name{Sys} \ \lor \ \name{pauseButton}~\name{Sys} \big)\big) \\
  \LTLglobally \big( \phantom{\name{ton}}\name{leaveApp}~\name{Sys}  & \to \big(\neg \, U_{\name{play}} \, \LTLweakuntil \ \name{resumeApp}~\name{Sys}\big) \hspace{38pt} \big)\\
   \LTLglobally \big( \name{pauseButton}~\name{Sys}                   & \to \big(\neg \, U_{\name{play}} \, \LTLweakuntil \ \name{playButton}~\name{Sys} \big) \hspace{33.6pt} \big)
\end{align}
The last two parts of the specification were already introduced in the motivating example.
If the music is playing, the music should pause when leaving the app and start playing again when returning to the app.
\begin{align}
  \begin{split}
    & \LTLglobally \big( \name{musicPlaying}~\name{MP} \ \wedge \ \name{leaveApp}~\name{Sys} \to \, U_{\name{pause}} \big)
  \end{split} \\
  \begin{split} \label{eq:assoonas}
    &\LTLglobally \big( \name{musicPlaying}~\name{MP} \ \wedge \ \name{leaveApp}~\name{Sys} \to  \\
     & \qquad \qquad \big((\name{pauseButton}~\name{Sys} \lor U_{\name{play}})\, \mathop{\mathcal{A}}\, \name{resumeApp}~\name{Sys} \big) \big)
  \end{split}
\end{align}
Note that in contrast to \cref{sec:motiv} we also test for the $ \name{event}(\name{pauseButton}) $ in guarantee~\cref{eq:assoonas}. This is necessary, since the specification would be unrealizable otherwise, as revealed by our synthesis tool. Indeed, the user may be smart enough to immediately pause the music when resuming the app, in which case it should not be played again.

\subsection{Proof of \cref{thm:tsl2ltl}}
\label[appendix]{proof:tsl2ltl}

\begin{proof}
  Assume $ \varphi_{\textit{LTL}} $ is realizable. Then there exists a
  winning strategy $ \sigma \from (2^{\pterms})^{+} \to 2^{\utermsp} $
  for the system player in the underlying LTL realizability
  game. Assume for the sake of contradiction that
  $ \varphi_{\textit{TSL}} $ is not realizable. Then there exists in
  input~$ \iota \in \inputs^{\omega} $ and a function
  assignment~$ \assign{\cdot} \from \fterms \to \functions $ such that
  for all $ \kappa \from \inputs^{+} \to \comps $ we have that
  $ \branch{\kappa}{\iota}, \iota, \nsats \varphi_{\textit{TSL}} $. We
  inductively construct the input
  sequence~$ \nu \in (2^{\pterms})^{\omega} $ and
  the computation~$ \comp \in \comps^{\omega} $ over $ t \in \dtime $ as follows:
  \begin{equation*}
    \begin{array}{rcl}
      \nu(t) & = & \set{ \pterm \in \pterms \mid \eval(\comp, \iota, t, \pterm) } \\[0.5em]
      \comp(t)(\name{s}) & = & \fterm \text{, \quad  where } \fterm \text{ is the unique element} \\
             && \mbox{\ } \hspace{4em} \text{such that } \upd{\name{s}}{\fterm} \in \sigma(\nu(0) \nu(1) \ldots \nu(t))
    \end{array}
  \end{equation*}
  Note that $ \fterm $ must be unique, due to the additional constraint
  \begin{equation*}
    \LTLglobally \Big ( \bigwedge_{\name{s} \in \onames} \, \bigvee_{\term \in \utermsp^{\hspace{0.5pt}\name{s}}}
    \big( \term \; \wedge \bigwedge_{\term' \in
      \utermsp^{\hspace{0.5pt}\name{s}} \setminus
      \set{ \term }} \neg \, \term' \big)  \Big)
  \end{equation*}
  introduced in the underapproximation. Furthermore, also note that
  $ \nu $ and $ \comp $ are well-defined, since
  $ \eval(\comp, \iota, t, \pterm) $ only considers values of
  $ \comp $ at previous times~\mbox{$ t' < t $}.
  Since~$ \varphi_{\textit{LTL}} $ is realizable, we have that
  $ \branch{\sigma}{\nu}, \nu \vDash \varphi_{\textit{LTL}} $\footnote{For a
    suitable definition for the semantics of $ \vDash $ consider for
    example~\cite{Schewe:2013}.}, but at the same time
  \mbox{$ \comp, \iota \nsats \varphi_{\textit{TSL}} $} by the
  unrealizability of $ \varphi_{\textit{TSL}} $. We show that this is
  contradictory via a structural induction over the structure
  of~$ \varphi_{\textit{TSL}} $ for all $ t \in \dtime $:
  \begin{itemize}

  \item Case $ \varphi_{\textit{TSL}} = \pterm $:
    \begin{align*}
      & \comp, \iota, t \sats \pterm \\ \Leftrightarrow~ &
      \eval(\comp, \iota, t, \pterm) \\ \Leftrightarrow~ &
      \pterm \in \nu(t) \\ \Leftrightarrow~ &
      \branch{\sigma}{\nu}, \nu, t \vDash \textsc{SyntacticConversion}(\pterm)
    \end{align*}

  \item Case $ \varphi_{\textit{TSL}} = \upd{\name{s}}{\fterm} $:
    \begin{align*}
      & \comp, \iota, t \sats \upd{\name{s}}{\fterm} \\ \Leftrightarrow~&
      \comp(t)(\name{s}) = \fterm \\ \Leftrightarrow~&
      \upd{\name{s}}{\fterm} \in \sigma(\nu(0) \nu(1) \ldots \nu(t)) \\ \Leftrightarrow~&
      \branch{\sigma}{\nu}, \nu, t \vDash \textsc{SyntacticConversion}(\upd{\name{s}}{\fterm})
    \end{align*}

  \item Case $ \varphi_{\textit{TSL}} = \neg \psi $:
    \begin{align*}
      & \comp, \iota, t \sats \neg \psi \\ \Leftrightarrow~ &
      \comp, \iota, t \nsats \psi \\ \overset{IH}{\Leftrightarrow}~ &
      \branch{\sigma}{\nu}, \nu, t \nvDash \textsc{SyntacticConversion}(\psi) \\ \Leftrightarrow~&
      \branch{\sigma}{\nu}, \nu, t \vDash \textsc{SyntacticConversion}(\neg \psi)
    \end{align*}

  \item Case $ \varphi_{\textit{TSL}} = \vartheta \wedge \psi $:
    \begin{align*}
      & \comp, \iota, t \sats \vartheta \wedge \psi \\ \Leftrightarrow~ &
      \comp, \iota, t \sats \vartheta \ \wedge \ \comp, \iota, t \sats \psi \\ \overset{IH}{\Leftrightarrow}~ &
      \branch{\sigma}{\nu}, \nu, t \vDash \textsc{SyntacticConversion}(\vartheta) \ \wedge \branch{\sigma}{\nu}, \nu, t \vDash \textsc{SyntacticConversion}(\psi) \\ \Leftrightarrow~&
      \branch{\sigma}{\nu}, \nu, t \vDash \textsc{SyntacticConversion}(\vartheta \wedge \psi)
    \end{align*}

  \item Case $ \varphi_{\textit{TSL}} = \LTLnext \psi $:
    \begin{align*}
      & \comp, \iota, t \sats \LTLnext \psi \\ \Leftrightarrow~ &
      \comp, \iota, t+1 \sats \psi \\ \overset{IH}{\Leftrightarrow}~ &
      \branch{\sigma}{\nu}, \nu, t+1 \vDash \textsc{SyntacticConversion}(\psi) \\ \Leftrightarrow~&
      \branch{\sigma}{\nu}, \nu, t \vDash \textsc{SyntacticConversion}(\LTLnext \psi)
    \end{align*}

  \item Case $ \varphi_{\textit{TSL}} = \vartheta \LTLuntil \psi $:
    \begin{align*}
      & \comp, \iota, t \sats \vartheta \LTLuntil \psi \\ \Leftrightarrow~ &
\exists t'' \geq t. \ \
         \forall t \leq t' < t''. \ \ \comp, \iota, t' \sats \vartheta \ \,
         \wedge \ \, \comp, \iota, t'' \sats \psi                                                                             \\ \overset{IH}{\Leftrightarrow}~ &
\exists t'' \geq t. \ \
         \forall t \leq t' < t''. \\ & \mbox{\qquad} \branch{\sigma}{\nu}, \nu, t' \vDash \textsc{SyntacticConversion}(\vartheta) \ \,
         \wedge \\ & \mbox{\qquad} \branch{\sigma}{\nu}, \nu, t'' \vDash \textsc{SyntacticConversion}(\psi)                                                                             \\ \Leftrightarrow~ &
      \branch{\sigma}{\nu}, \nu, t \vDash \textsc{SyntacticConversion}(\vartheta \LTLuntil \psi) \hfill \qedhere
    \end{align*}

  \end{itemize}
\end{proof}

\subsection{Extended proof of \cref{thm:decidability}}
\label[appendix]{proof:decidability}

\begin{proof}
  We give a reduction from the Post Correspondence Problem~(PCP).
  Given two finite lists $ w_{0}w_{1}\ldots w_{n} $ and
  $ v_{0}v_{1}\ldots v_{n} $ of equal length, each containing $ n $ finite
  words over some finite alphabet~$ \Sigma $, is there some finite
  sequence $ i_{0}i_{1}\ldots i_{k} \in \nats^{*} $ such that
  $ w_{i_{0}}w_{i_{1}}\ldots w_{i_{k}} = v_{i_{0}} v_{i_{1}} \ldots
  v_{i_{k}} $. The problem is undecidable~\cite{post1946}.
  We reduce PCP to the realizability question of TSL, where we
  translate an arbitrary instance of PCP to a TSL formula~$ \varphi $,
  which is realizable if and only if there is a solution to the PCP
  instance. To this end, let $ n \in \nats $,
  $ w_{0}w_{1}\ldots w_{n} \in \Sigma^{*} $ and
  $ v_{0}v_{1}\ldots v_{n} \in \Sigma^{*} $ be given.  We fix
  $ \pnames = \set{ \name{p} } $ for some unary predicate~$ p $,
  \mbox{$ \fnames = \pnames \cup \Sigma \cup \set{ \name{X} } $}, where every
  $ \name{f} \in \Sigma $ corresponds to a unary function and
  $ \name{X} $ corresponds to a $ 0 $-nary function,
  $ \inames = \emptyset $, $ \onames = \set{ \name{A}, \name{B} } $, and
  $ \cells = \onames $. We define~$ \varphi $ via:
  \begin{equation*}
    \begin{array}{l}
    \varphi =  \Big( \, \upd{\name{A}}{\const{X}} \ \wedge \ \upd{\name{B}}{\const{X}} \, \Big)
    \ \wedge \\[0.2em] \phantom{\varphi \ =} \LTLnext\, \LTLglobally
    \Big( \bigvee_{j=0}^{n} \big( \upd{\name{A}}{\mu(w_{j},\name{A})} \, \wedge \,
    \upd{\name{B}}{\mu(v_{j},\name{B})} \big) \Big) \ \wedge  \\[0.2em]
      \phantom{\varphi \ =} \LTLnext \, \LTLnext \, \LTLfinally \Big( \, \name{p}~\name{A}
    \, \leftrightarrow \, \name{p}~\name{B} \, \Big)
    \end{array}
  \end{equation*}
  where
  $ \mu(x_{0}x_{1}\ldots x_{m},s) = x_{0}~(x_{1}~(\ldots (x_{m}~s) \ldots
  )) $.

  Intuitively, we first assign the signals $ \name{A} $ and
  $ \name{B} $ a constant base value. Then, from the next time step
  on, we have to pick pairs $ (w_{j},v_{j}) $ in every time step. Our
  choice is stored in the signals~$ \name{A} $ and $ \name{B} $,
  respectively. Finally, we check that the constructed sequences of
  function applications are equal at some point in time, where we use the
  universally quantified predicate~$ \name{p} $ to check for equality.

  \medskip

  \noindent The TSL formula~$ \varphi $ is realizable if and only if
  there is an index sequence~$ i_{0}i_{1}\ldots i_{k} $ such that
  $ w_{i_{0}}w_{i_{1}}\ldots w_{i_{k}} = v_{i_{0}} v_{i_{1}} \ldots
  v_{i_{k}} $:

  \smallskip

  ``$ \Rightarrow $'': Assume that $ \varphi $ is realizable, i.e.,
  there is some strategy~$ \sigma \from \inputs^{*} \to \comp $ that
  satisfies $ \varphi $ for
  $ \iota = \emptyset^{\hspace{0.5pt}\omega} $ and all possible
  choices of $ \assign{\cdot} \from \fnames \to \functions $. We fix
  $ \assign{\textit{init}_{\name{A}}} =
  \assign{\textit{init}_{\name{B}}} = \assign{\name{X}} = \varepsilon
  $ and $ \assign{x} \from \Sigma^{*} \to \Sigma^{*} $ with
  $ \assign{x}~w = wx $ for all $ w \in \Sigma^{*} $ and
  $ x \in \Sigma $. We do not fix any assignment to $ \name{p}
  $. Nevertheless, by the realizability of~$ \varphi $, there is a
  position~$ m > 1 $ at which
  $ \name{p}~\name{A} \leftrightarrow \name{p}~\name{A} $ is
  satisfied, independent of the predicate assigned to $ \name{p} $. We
  obtain that
  $ \eval(\branch{\sigma}{\iota},\iota,m,\name{A}) =
  \eval(\branch{\sigma}{\iota},\iota,m,\name{B}) $, \linebreak since
  otherwise there would be a predicate that detects the difference. As
  there is no other influence
  on~$ \branch{\sigma}{\iota} \in \comps^{\omega} $, depending on the
  choice of $ \name{p} $, we obtain that the semantics of
  $ \varrho_{\hspace{-1pt}\langle \hspace{-1pt}\cdot
    \hspace{-1pt}\rangle \hspace{-1pt}, \comp, \iota} \in
  \outputs^{\hspace{0.5pt}\omega} $
  \begin{eqnarray*}
   \varrho_{\hspace{-1pt}\langle \hspace{-1pt}\cdot
    \hspace{-1pt}\rangle \hspace{-1pt}, \comp, \iota} & = & \set{ \varrho_{\hspace{-1pt}\langle \hspace{-1pt}\cdot
    \hspace{-1pt}\rangle \hspace{-1pt}, \comp, \iota}(0)(\name{A}) \mapsto a_{0}, \,
                                                             \varrho_{\hspace{-1pt}\langle \hspace{-1pt}\cdot
    \hspace{-1pt}\rangle \hspace{-1pt}, \comp, \iota}(0)(\name{B}) \mapsto b_{0} } \\
    & & \set{ \varrho_{\hspace{-1pt}\langle \hspace{-1pt}\cdot
    \hspace{-1pt}\rangle \hspace{-1pt}, \comp, \iota}(1)(\name{A})
    \mapsto a_{1}, \, \varrho_{\hspace{-1pt}\langle \hspace{-1pt}\cdot
    \hspace{-1pt}\rangle \hspace{-1pt}, \comp, \iota}(1)(\name{B}) \mapsto b_{1} } \\
    & & \ldots
  \end{eqnarray*}
  are well defined (even without fixing $ \name{p} $) and induce the
  sequences $ a_{0}a_{1}\ldots \in \Sigma^{\omega} $ and
  $ b_{0}b_{1}\ldots \in \Sigma^{\omega} $. First, we observe that
  $ a_{0} = a_{1} = b_{0} = b_{1} = \varepsilon $ by the definition of
  $ \varphi $ and the choice of the initialization functions and
  $ \assign{\cdot} $. From the choice of $ \assign{\cdot} $, we also
  obtain that $ a_{t} = a_{t-1}w_{i_{t}} $ and
  $ b_{t} = b_{t-1}v_{i_{t}} $ for every $ t > 1 $ and some
  $ 0 \leq i_{t} \leq n $. A simple induction shows that every
  $ a_{t} = w_{i_{2}}w_{i_{3}}\ldots w_{i_{t}} $ and
  $ b_{t} = v_{i_{2}}v_{i_{3}} \ldots v_{i_{t}} $ for every $ t > 1
  $. It follows that
  $ w_{i_{2}}w_{i_{3}}\ldots w_{i_{m-1}} = v_{i_{2}}v_{i_{3}} \ldots
  v_{i_{m}} $ from equality of $ a_{m} $ and $ b_{m} $ at
  position~$ m $. This concludes this part of the proof, since we have
  that $ i_{2}i_{3}\ldots i_{m} $ is a solution for the PCP instance.

  \smallskip

  ``$ \Leftarrow $'': Now, assume that there is a
  solution~$ i_{0}i_{1}\ldots i_{k} $ to the PCP instance.
  Since~$ \inames = \emptyset $ it suffices to
  construct the computation~$ \comp = c_{0}c_{1}\ldots $ with
  $ c_{0}(\name{A}) = c_{0}(\name{B}) = \const{X} $ and for all
  $ t > 0 $:
  $ c_{t}(\name{A}) = \mu(w_{(t-1) \!\!\mod (k+1)}, \name{A}) $ and
  $ c_{t}(\name{B}) = \mu(w_{(t-1) \!\!\mod (k+1)}, \name{B}) $. It is
  straightforward to see that $ \comp $ satisfies
  $ \mbox{\ensuremath{\upd{\name{A}}{\name{X}}}} $,
  $ \mbox{\ensuremath{\upd{\name{B}}{\name{X}}}}$ and
  \begin{equation*}
     \LTLnext \, \LTLglobally \, \Big( \bigvee_{j=0}^{n}
    \big( \upd{\name{A}}{\mu(w_{j},\name{A})} \, \wedge \,
    \upd{\name{B}}{\mu(v_{j},\name{B})} \big) \Big).
  \end{equation*}
  Thus, it just remains to argue that $ \comp $ satisfies
  $ \LTLnext \, \LTLnext \, \LTLfinally ( \name{p}~\name{A} \,
  \leftrightarrow \, \name{p}~\name{B} ) $. To this end, let
  $ j_{0}j_{1}\ldots = (i_{0}i_{1}\ldots i_{k})^{\omega} $. Then a
  simple induction shows that
  \begin{equation*}
    \varrho_{\hspace{-1pt}\langle \hspace{-1pt}\cdot
      \hspace{-1pt}\rangle \hspace{-1pt}, \comp, \iota}
    (t)(\name{A}) =
    \eval(\comp,\iota,t,\mu(w_{j_{0}}w_{j_{1}}\ldots w_{j_{t-2}},\name{X}))
  \end{equation*}
  and
  $ \varrho_{\hspace{-1pt}\langle \hspace{-1pt}\cdot
    \hspace{-1pt}\rangle \hspace{-1pt}, \comp, \iota}(t)(\name{B}) =
  \eval(\comp,\iota, t,\mu(v_{j_{0}}v_{j_{1}}\ldots
  v_{j_{t-2}},\name{X})) $ for all $ t > 1 $,
  $ \iota = \emptyset^{\hspace{0.5pt}\omega} $, and all choices of
  $ \assign{\cdot} $. Now, consider that especially for $ t = k+2 $ we
  have that
  \begin{eqnarray*}
    w_{j_{0}}w_{j_{1}}\ldots w_{j_{t-2}} & =
    & w_{i_{0}}w_{i_{1}}\ldots  w_{i_{k}} \\
    & = & v_{i_{0}}v_{i_{1}}\ldots v_{i_{k}} \\
    & = &  v_{j_{0}}v_{j_{1}}\ldots v_{j_{t-2}}
  \end{eqnarray*}
  and, thus, also
  $  \varrho_{\hspace{-1pt}\langle \hspace{-1pt}\cdot
    \hspace{-1pt}\rangle \hspace{-1pt}, \comp, \iota}(k+2)(\name{A}) =  \varrho_{\hspace{-1pt}\langle \hspace{-1pt}\cdot
    \hspace{-1pt}\rangle \hspace{-1pt}, \comp, \iota}(k+2)(\name{B}) $,
  independent of the choice of $ \assign{\cdot} $. As this implies
  that  $ p(\varrho_{\hspace{-1pt}\langle \hspace{-1pt}\cdot
    \hspace{-1pt}\rangle \hspace{-1pt}, \comp, \iota}(k+2)(\name{A})) = p(\varrho_{\hspace{-1pt}\langle \hspace{-1pt}\cdot
    \hspace{-1pt}\rangle \hspace{-1pt}, \comp, \iota}(k+2)(\name{B}))
  $ for any unary predicate~$ p \in \predicates $, it proves that
  $ \comp $ satisfies
  $ \name{p}~\name{A} \, \leftrightarrow \, \name{p}~\name{B} $ at
  position $ k+2 $. Hence, the computation~$ \comp $ also satisfies
  $ \LTLnext \, \LTLnext \, \LTLfinally ( \name{p}~\name{A} \,
  \leftrightarrow \, \name{p}~\name{B} ) $, which concludes the
  proof.
\end{proof}
\noindent Note that we do not use any inputs to encode the PCP instance in the proof
above. As a consequence, it immediately follows that the
satisfiability problem of \TSL is undecidable as well.

\subsection{Other Benchmarks}
\label[appendix]{apx:benchmarks}

\subsubsection{Escalator}

As an illustrative example of the logic, we build up the
specification of an escalator controller. To this end, we first commit
to the following physical model.
To interact with the environment the escalator is equipped with a
motor, which moves its steps either up, down, or is turned off. To
observe the environment, it additionally can read from two sensors, at
the bottom and at the top, that reliably detect whenever somebody
enters or exists. To program the controller, we receive the inputs of
the sensors via two input signals: $ \name{bottom} $ and
$ \name{top} $. The steps are controlled via the output signal:
$ \name{steps} $.
We start with a simple, non-reactive version that continuously moves
up. This behavior is described via:
\begin{equation*}
  \LTLglobally \; \upd{\name{steps}}{\name{MOVEUP}()}
\end{equation*}
Note that we use $ \name{MOVEUP}() $ as a $ 0 $-nary function here, as
the command does not depend on any other given signal. However, for a
concrete system, we instead would replace it with the matching library
call (including possible static parameters), that passes the movement
command to the motor's driver. Nevertheless, reflecting these details here
would go beyond the scope of this illustration.

Next, let us make our escalator reactive, in the sense that it only
moves up, if there is actually somebody on it, using it.
A first specification looks as follows:
\begin{equation*}
  \begin{array}{l}
    \LTLglobally \Big( \big( (\name{enterEvent}~\name{bottom} \wedge \neg \name{exitEvent}~\name{top})
    \leftrightarrow \upd{\name{steps}}{\name{MOVEUP}()} \big) \wedge \mbox{\ } \\
    \phantom{\LTLglobally \Big(} \big( (\name{exitEvent}~\name{top} \wedge \neg \name{enterEvent}~\name{bottom})
    \leftrightarrow \upd{\name{steps}}{\name{STOP}()} \big) \Big)
  \end{array}
\end{equation*}
While this specification is realizable, it also turns out to be
incomplete. The elevator stops as soon as the first user leaves it at
the top, but there still may be other users behind.

A short inspection reveals that we cannot solve this problem via
purely reacting to the given inputs. Instead, we need to introduce a
counter (used as an internal signal), which keeps track of the users
on the escalator. Hence, we change the previous specification into
\begin{equation*}
  \begin{array}{l}
    \LTLglobally \Big( \big( (\name{enterEvent}~\name{bottom} \wedge \neg \name{exitEvent}~\name{top}) \leftrightarrow \upd{\name{users}}{\name{(+)}~\name{users}~\name{1}} \big) \wedge \mbox{\ } \\
    \phantom{\LTLglobally \Big(} \big( (\name{exitEvent}~\name{top} \wedge \neg \name{enterEvent}~\name{bottom})
    \leftrightarrow \upd{\name{users}}{ \name{(+)}~\name{users}~(\name{-1})} \big) \Big)
  \end{array}
\end{equation*}
With this change, it only remains to start and stop the escalator
whenever the number of users toggles between zero and non-zero:
\begin{equation*}
  \begin{array}{l}
    \LTLglobally \Big( \big( (\name{(==)}~\name{users}~\name{0} \ \wedge \ \LTLnext \neg \, (\name{(==)}~\name{users}~\name{0})) \leftrightarrow \LTLnext \, \upd{\name{steps}}{\name{MOVEUP}()} \big) \\
    \phantom{\LTLglobally \Big( } \big( ( \neg (\name{(==)}~\name{users}~\name{0}) \wedge \LTLnext \, (\name{(==)}~\name{users}~\name{0})) \leftrightarrow \LTLnext \; \upd{\name{steps}}{\name{STOP}()} \ \ \ \big) \Big)
  \end{array}
\end{equation*}
In conjunction with the previous part, we obtain a complete
description of a reactive escalator, which we can use to synthesize a
respective FPP program. The result not only satisfies the
specification, but also is immediately executable on the controller
(after compilation).

Note that we can easily extend our specification, by adding even more
properties. This is possible without changing the previous parts at
all.  For example, consider an alarm, that is activated whenever there
are too many users on the escalator.  Another variant would be a smart
version, which moves up and down, triggered by the entrance point of
the next user whenever the system is idle. For the sake of
illustration of these statements, we give a possible realization of
the second variant in the sequel.

We present the specification of a smart escalator, which is able to
move into both directions: up and down. Thereby, the direction is
determined by the entrance point of the first user entering the
escalator when it is empty. If the escalator already moves into a
specific direction, we ignore enter and exit events into the opposite
direction until it stopped again. The specification
$ \varphi = \LTLglobally \psi \rightarrow \bigwedge_{j=0}^{7}
\LTLglobally \vartheta_{j} $ consists of:
\begin{align*}
  \psi = \ & \neg \big(\name{(==)}~\name{steps}~\name{MOVEDOWN}() \ \wedge \ \name{(==)}~\name{steps}~\name{MOVEUP}()\big) \wedge \\
           & \neg \big(\name{enterEvent}~\name{top} \wedge  \name{exitEvent}~\name{top} \big) \wedge \mbox{\,}\\
           & \neg \big(\name{enterEvent}~\name{bottom} \wedge  \name{exitEvent}~\name{bottom} \big) \\[0.4em]
  \vartheta_{0} =\ & \big( (\name{enterEvent}~\name{bottom} \wedge \neg \, \name{exitEvent}~\name{top} \wedge \mbox{\ }
                     \neg (\name{(==)}~\name{steps}~\name{MOVEDOWN}()) ) \vee \mbox{ \ }\\
                   & \phantom{\big(} (\name{enterEvent}~\name{top} \wedge \neg \, \name{exitEvent}~\name{bottom} \wedge \mbox{\ }
                     \neg (\name{(==)}~\name{steps}~\name{MOVEUP}())) \big) \\
                   & \leftrightarrow \upd{\name{users}}{\name{(+)}~\name{users}~\name{1}}  \\
  \vartheta_{1} =\ & \big((\name{exitEvent}~\name{top} \wedge \neg \, \name{enterEvent}~\name{bottom} \wedge \mbox{\ }
                     \neg (\name{(==)}~\name{steps}~\name{MOVEDOWN}())) \vee \mbox{\ } \\
                   & \phantom{\big(}(\name{exitEvent}~\name{bottom} \wedge \neg \, \name{enterEvent}~\name{top} \wedge \mbox{\ }
                     \neg (\name{(==)}~\name{steps}~\name{MOVEUP}()))\big) \\
                   & \leftrightarrow \upd{\name{users}}{\name{(+)}~\name{users}~\name{(-1)}}  \\[0.4em]
  \vartheta_{2} =\ & \upd{\name{steps}}{\name{MOVEUP}()} \\
                   & \rightarrow \ \name{(==)}~\name{users}~\name{0} \; \wedge \; \name{enterEvent}~\name{bottom} \\[0.4em]
  \vartheta_{3} =\ & \upd{\name{steps}}{\name{MOVEDOWN}()} \\
                   & \rightarrow \ \name{(==)}~\name{users}~\name{0} \; \wedge \; \name{enterEvent}~\name{top} \\[0.4em]
  \vartheta_{4} =\ & \name{(==)}~\name{users}~\name{0} \; \wedge \; \name{enterEvent}~\name{bottom} \wedge \neg \name{enterEvent}~\name{top} \\
                   & \rightarrow \upd{\name{steps}}{\name{MOVEUP}()}  \\[0.4em]
  \vartheta_{5} =\ & \name{(==)}~\name{users}~\name{0} \; \wedge \; \name{enterEvent}~\name{top} \wedge \mbox{\ } \\
                   & \neg \name{enterEvent}~\name{bottom} \rightarrow \upd{\name{steps}}{\name{MOVEDOWN}()}  \\[0.4em]
  \vartheta_{6} =\ & \name{(==)}~\name{users}~\name{0} \; \wedge \; \name{enterEvent}~\name{top} \wedge \name{enterEvent}~\name{bottom} \\
                   & \rightarrow \big(\upd{\name{steps}}{\name{MOVEUP}()} \vee \upd{\name{steps}}{\name{MOVEDOWN}()} \big)  \\[0.4em]
  \vartheta_{7} =\ & \neg (\name{(==)}~\name{users}~\name{0}) \wedge \LTLnext \! \big(\name{(==)}~\name{users}~\name{0} \wedge
                     \neg \name{enterEvent}~\name{top} \wedge \neg \name{enterEvent}~\name{bottom}\big) \\
                   & \leftrightarrow \upd{\name{steps}}{\name{STOP}} \\
\end{align*}

\subsubsection{FRPZoo}

A special role is played by the
FRPZoo benchmark set, which refers to a standard online benchmark
suite, designed to compare FRP library language designs~\cite{FRPzoo}.
The specification separates between three different behaviors, given
as scenarios 0, 5, and 10. In every scenario, two buttons can be
clicked: a ``click count'' button, which counts the number of clicks,
and a ``toggle'' button, which toggles the enable/disable state of the ``click
count'' button. The value of the counter is displayed via some
output interface. The three scenarios differ with respect to the exact
conditions of when the counter is updated, reset or displayed.

\subsubsection{Haskell-TORCS}

The Haskell-TORCS benchmarks synthesize controllers for an autonomous vehicle.
Our specifications build upon the example of the Haskell-TORCS bindings for building FRP controllers~\cite{SCAV2017} in The Open Race Car Simulator (TORCS)~\cite{torcs}. The bindings are also used to run the synthesized implementations within the simulator.
Autonomous vehicles use limited sensor data about the environment (\eg the distance to nearest obstacle) to control actuators in the car (\eg the steering wheel).
The Haskell-TORCS set of benchmarks synthesize a controller from \TSL specifications where the sensors and actuators are the input and output signals respectively.
The functions used in the \TSL specifications for the Haskell-TORCS benchmarks, for example ``slowDown'' or ``turnLeft'', are implemented after the controller synthesis process.
In this way, we obtain a guarantee on the larger behavior of the system, while still allowing numerically sensitive, data level manipulations, to be optimized as required by the application.

The first ``simple'' Haskell-TORCS controller combines simple functions without states.
The ``advanced'' controllers included more detailed planning behavior when approaching a turn.
The specifications are also modular, in the sense that control of the steering wheel and control of the gears are given separate specifications, and combined into a single FRP program after synthesis.

\subsubsection{Counters}

The benchmarks of \cref{table:results2} consider a user interface that
allows to increment either one or two counters while at the same time
ensuring that each counter stays in range, as restricted by
predicate. The benchmarks are inspired by examples of the Reactive
Banana FRP library~\cite{reactivebanana}, which can be found under
\url{https://wiki.haskell.org/Reactive-banana/Examples}. The
benchmarks work on a small interface, over which the user can
increment or decrement one or two counters.In the two counter case,
he or she can switch between the counters using an additional toggling
input.

In the simple variant, the system only must ensure that each counter
stays in range, where it is allowed to ignore an increment or
decrement request from the user, if this could produce a counter value
out of range. In the more advanced variant, we also assume the
existence of a graphical interface, where the inputs are realized via
pressing different buttons. Here, the synthesizer can only avoid an
increment of the user by disabling a button before pressing it could
lead to an out of range update. For all benchmarks, the synthesizer
must utilize purity of the increment / decrement operation to test
early enough whether the environment could produce an out of range
value to prevent this either by ignoring the input or by disabling the
corresponding buttons.

\end{document}